\pdfoutput=1
\RequirePackage{ifpdf}
\ifpdf 
\documentclass[pdftex]{sigma}
\else
\documentclass{sigma}
\fi

\newcommand{\adg}{{\rm ad}_{\mathfrak{g}}}

\newtheorem{Theorem}{Theorem}[section]
\newtheorem*{Theorem*}{Theorem}

\newtheorem{Lemma}[Theorem]{Lemma}

 { \theoremstyle{definition}

 }

\numberwithin{equation}{section}

\begin{document}

\allowdisplaybreaks

\newcommand{\arXivNumber}{2303.14181}

\renewcommand{\PaperNumber}{031}

\FirstPageHeading

\ShortArticleName{Transformations of Currents in Sigma-Models with Target Space Supersymmetry}

\ArticleName{Transformations of Currents in Sigma-Models\\ with Target Space Supersymmetry}

\Author{Vinicius BERNARDES~$^{\rm a}$, Andrei MIKHAILOV~$^{\rm b}$ and Eggon VIANA~$^{\rm b}$}

\AuthorNameForHeading{V.~Bernardes, A.~Mikhailov and E.~Viana}

\Address{$^{\rm a)}$~CEICO, Institute of Physics of the Czech Academy of Sciences,\\
\hphantom{$^{\rm a)}$}~Na Slovance~2, 182 00 Prague~8, Czech Republic}
\EmailD{\href{mailto:vb.silva@unesp.br}{vb.silva@unesp.br}}

\Address{$^{\rm b)}$~Instituto de Fisica Teorica, Universidade Estadual Paulista,\\
\hphantom{$^{\rm b)}$}~R.~Dr.~Bento Teobaldo Ferraz 271, Bloco II -- Barra Funda,\\
\hphantom{$^{\rm b)}$}~CEP:01140-070 -- Sao Paulo, Brasil}
\EmailD{\href{mailto:a.mkhlv@gmail.com}{a.mkhlv@gmail.com}, \href{mailto:eggon.viana@unesp.br}{eggon.viana@unesp.br}}

\ArticleDates{Received April 19, 2023, in final form March 26, 2024; Published online April 10, 2024}

\Abstract{We develop a framework for systematic study of symmetry transformations of sigma-model currents in a special situation, when symmetries have a well-defined projection onto the target space. We then apply this formalism to pure spinor sigma-models, and describe the resulting geometric structures in the target space (which in our approach includes the pure spinor ghosts). We perform a detailed study of the transformation properties of currents, using the formalism of equivariant cohomology. We clarify the descent procedure for the ``universal'' deformation corresponding to changing the overall scale of the worldsheet action. We also study the contact terms in the OPE of BRST currents, and derive some relations between currents and vertex operators which perhaps have not been previously acknowledged. We also clarify the geometrical meaning of the ``minimalistic'' BV action for pure spinors in AdS.}


\Classification{83E30; 17B55; 55N25; 70S10}

\section{Introduction}

Consider a two-dimensional sigma-model invariant under some group of symmetries. Infinitesimal symmetries
form a Lie algebra $\mathfrak{g}$. For each $\xi\in\mathfrak{g}$, there is a corresponding infinitesimal field transformation $\delta_{\xi}$.
The Lagrangian is invariant up to a total derivative:
\begin{equation}
 \delta_{\xi} L = {\rm d}\alpha\langle\xi\rangle.
 \label{DeltaL}\end{equation}
We consider the Lagrangian as a two-form on the worldsheet;
$\alpha$ is a one-form.
We use angular brackets, to emphasize that $\alpha$ depends on $\xi$ \emph{linearly}.
Noether charge, which we denote $\mathfrak{Q}\langle\xi\rangle$, is given by an integral of a conserved current over a space-like contour on the
worldsheet
\begin{equation*}
 \mathfrak{Q}\langle\xi\rangle = \oint j\langle\xi\rangle.
 \end{equation*}
The charges transform covariantly up to constant cocycles
\begin{equation*}
 \delta_{\xi} \mathfrak{Q}\langle\eta\rangle = \mathfrak{Q}\langle[\xi,\eta]\rangle + C\langle\xi,\eta\rangle,
 \end{equation*}
where $C\langle\xi,\eta\rangle$ is constant (i.e., field-independent).
In many cases, it is true that $C\langle\xi,\eta\rangle = 0$, and we will here assume that this is the case.
Then, it follows that the currents transform covariantly up to a total derivative
\begin{equation}
 \delta_{\xi} j\langle\eta\rangle = j\langle[\xi,\eta]\rangle + {\rm d}\Psi\langle\xi,\eta\rangle.
 \label{DefPsi}\end{equation}
In this paper, we will investigate various general properties of $\alpha$ and $\Psi$, and the relation between them,
in one special case. This special case is when the symmetries are ``projectable to the target space'', in the sense
which we explain now.

\subsection{Symmetries projectable to the target space}

The phase space is usually the cotangent bundle over the configuration space. Symmetries act in the cotangent space.
In this paper, we will concentrate on a special case, when the symmetries have a well-defined
projection on the configuration space.
(This is not true in general, since the symmetry transformation of coordinates may depend on momenta.)
Moreover, we assume that an infinitesimal symmetry $\xi$ corresponds to a vector field $v\langle\xi\rangle$ on the target space.
(This would not be the case if derivatives of the fields were involved.)

In this special case, $\alpha\langle\xi\rangle$ of equation~(\ref{DeltaL}) can be viewed
as a one-form on the target space linearly dependent on $\xi$
as a parameter. Similarly, the $\Psi$ of equation~(\ref{DefPsi}) can be viewed as a
function on the target space parametrized by $\xi$ and $\eta$. We explain
that $\alpha$ and $\Psi$ should be thought of as defining the cochains
of the bicomplex ${\rm d} + {\rm d}_{\rm Lie}$, where ${\rm d}$ is the de Rham differential
in the target space $X$, and ${\rm d}_{\rm Lie}$ the differential of the Lie algebra
cohomology complex.
We derive several relations between $\alpha$, $\Psi$,
and various objects associated with them. The first observation is that there exists~$V$,
a 2-cochain of $\mathfrak{g}$ with values in $C^{\infty}(X)$, such that
\begin{align*} & ({\rm d}_{\rm Lie} + {\rm d}) (V - \alpha) =
 W - {\rm d}\alpha,
 \end{align*}
where $W$ is constant on the target space;
it defines a cohomology class $[W]\in H^3(\mathfrak{g},\mathbb{R})$.
See Section~\ref{NoetherCurrents}.

The relation between $\alpha$ and $\Psi$ is discussed in
Section~\ref{TransformationOfCurrents}.
We formulate it in terms of the ``BRST model'' of equivariant cohomology
(see \cite{Cordes:1994fc}):
\begin{equation*}
{\rm d}_{\rm BRST}(V - \alpha) = W + \Psi + \cdots.
\end{equation*}
Details are described in Section \ref{TransformationOfCurrents}.

As a first example, we consider the WZW model in Section \ref{WZW}.
Then we discuss applications to the pure spinor sigma-models.

\subsection{Applications to pure spinor superstring}

Our study is motivated mostly by the pure spinor
sigma-model of \cite{Berkovits:2001ue}, and it will be our main example.
We will restrict ourselves to Type~II string. The corresponding sigma-model, in any background,
has the following distinguishing property. There are always two anticommuting nilpotent symmetries,
$Q_L$ and $Q_R$:
\begin{equation*}
 \{Q_L,Q_L\} = \{Q_R,Q_R\} = \{Q_L,Q_R\} = 0
 \end{equation*}
and two $u(1)$ symmetries called ``ghost numbers''. The $Q_L$ and $Q_R$ have charges $(1,0)$ and $(0,1)$,
correspondingly, under these two symmetries.

\subsubsection{Universal vertex operator}

The Lagrangian is only BRST invariant up to a total derivative
\begin{gather*}
\begin{split}
& Q_LL=
 {\rm d}\alpha_L,
 \qquad Q_RL=
 {\rm d}\alpha_R,
 \qquad Q_L\alpha_L=
 {\rm d} V_{LL},
 \\
 & Q_L\alpha_R + Q_R\alpha_L=
 {\rm d} V_{LR},
 \qquad Q_R\alpha_R=
 {\rm d} V_{RR}.
\end{split}
 \end{gather*}
This defines the ``universal vertex operator''
\begin{gather*}
 V = V_{LL} + V_{LR} + V_{RR},
\qquad
 (Q_L + Q_R) V =0.
 \end{gather*}
It is the unintegrated vertex operator corresponding
to the overall scale of the worldsheet action. Its descent works off-shell on the worldsheet,
and has a general description in any curved background.

The case of flat space is special. In this case $V$ is BRST exact
\begin{equation}
 V = (\lambda_L\Gamma^m\theta_L) (\lambda_R\Gamma^m\theta_R) =
 - (Q_L + Q_R)\biggl(\frac{1}{2}(X^m [(\lambda_L\Gamma^m\theta_L) - (\lambda_R\Gamma^m\theta_R)])\biggr).
 \label{IntroFlatSpaceV}\end{equation}
Notice that $(\lambda_L\Gamma^m\theta_L) - (\lambda_R\Gamma^m\theta_R)$ is the ghost number one vertex
operator corresponding to the translation.
Generally speaking, exact vertices
$Q(v^m(X)[(\lambda_L\Gamma^m\theta_L) - (\lambda_R\Gamma^m\theta_R)])$
correspond to infinitesimal target space diffeomorphisms.
In particular, the $V$ of equation~(\ref{IntroFlatSpaceV}) correspond to space-time dilatations.
Although the flat space $V$ is BRST exact, there is still an obstacle to being able to choose
a BRST-invariant Lagrangian. It is a nontrivial element of $H^1\bigl(Q,\Omega^1(X)\bigr)$ -- the
ghost number one BRST cohomology in one-forms on the target space.
While in general background the obstacle is in $H^2(Q,C^{\infty}(X))$, in flat space it is in
$H^1\bigl(Q,\Omega^1(X)\bigr)$. See Section~\ref{PureSpinorFlat}.

\subsubsection{BRST invariance of BRST currents}

While the Lagrangian is only invariant up to a total derivative, the BRST currents are
BRST closed strictly (not just up to a total derivative).
Moreover, they are actually BRST exact.
We will give a concise proof of this in
Section \ref{UniversalVertex}.
This is a consequence of the ghost number symmetry, which acts on fields as a phase rotation.

\subsubsection{Non-covariance of global symmetry currents}

Generally speaking, currents transform covariantly
only up to a total derivative. As a simple example, consider superstring in flat space in
pure spinor or Green--Schwarz formulation. The supersymmetry currents are, schematically
\begin{align*} &S_L =
 \theta_L\theta_L p_L + \partial_+ X\theta_L,
 \end{align*}
we have $\delta_L S_L \simeq \partial x$ and $\delta_R S_R \simeq \bar{\partial} x$.
This is not covariant. (The covariant transformation would have $P = *{\rm d}x$ on the right-hand side.)
The difference is in the total derivative ${\rm d}x$. There is a~cohomology class responsible
for this non-covariance, equation~(\ref{WInFlatSpace}).

\subsubsection{Contact terms}

The BRST currents $j\langle Q_L\rangle$ and $j\langle Q_R\rangle$
are holomorphic and antiholomorphic, respectively
\begin{align*}& j\langle Q_L\rangle =
 j^L_z {\rm d}z,
 \qquad j\langle Q_R\rangle =
 j^R_{\bar{z}} {\rm d}\bar{z}.
 \end{align*}
In Section \ref{ContactTerms}, we study the
contact term in the OPE of $j_L$ and $j_R$ and formulate a conjecture
about its relation to the universal vertex operator.

\subsubsection{Vertex operators corresponding to currents}

In Sections \ref{PureSpinorFlat}
and \ref{PureSpinorAdS}, we apply
our formalism to pure spinor sigma-models in flat space and~AdS.
Global symmetry currents in~AdS are invariant under $\mathfrak{psu}(2,2|4)$.
To each global current~$j_a$ corresponds a ghost number one vertex operator
$\Lambda_a$. Our formalism implies some relations between them. We also clarify the
geometrical meaning of the ``minimalistic'' B-field of~\cite{Mikhailov:2017mdo}.

In Section~\ref{FlatSpaceLimitOfAdS},
we discuss some subtleties in taking the flat space limit of AdS.

\subsection{Previous work}

There have been many studies of currents in sigma-model, in particular following \cite{Alekseev:2004np}.
We would especially mention \cite{Suszek:2018bvx,Suszek:2019cum,Zavaleta:2019dop}, which have some overlap with this work.
Our treatment of the worldsheet symplectic structure and global symmetries is similar to
the one developed in greater generality in \cite{Deligne:1999qp} and \cite{Blohmann}.
It would be interesting to see if some results of our work could be extended to more
general case (dimension higher than two and symmetries not projectable to the target space)
using that more general approach.

\section{Noether currents}\label{NoetherCurrents}

\subsection{Symmetries projectable onto the target space}

Consider a classical field theory.
Its phase space is $T^*F$ where $F$ is the space of field configurations.
We assume that our theory is a sigma-model. This means that field configurations are maps
from the worldsheet $\Sigma$ to the target space $X$:
\begin{equation*}
 F = \mbox{Maps}(\Sigma,X).
 \end{equation*}
Suppose that the theory is invariant under a Lie group $G$, and the action of $G$ comes
from the action on the target space $X$.

\subsection{The structure of Noether currents}

\subsubsection{Symplectic structure on the phase space of classical sigma-model}

The \emph{symplectic potential} $\vartheta$ is a one-form on the phase space and a one-form on the worldsheet,
which is defined by considering the variation of the Lagrangian. Let us pick a solution
of classical equations of motion, and consider its infinitesimal deformation $\delta$
\emph{as an off-shell field configuration}. This means that our variation does not have to
preserve the equations of motion,
but the original configuration which we are varying \emph{is} on-shell. Since the original
configuration is on-shell, the variation of the Lagrangian is a total derivative of
some 1-form on the worldsheet:%
\begin{equation*}
 \delta L = {\rm d}\vartheta.
 \end{equation*}
Here ${\rm d}$ is the de Rham differential on the worldsheet,
and $\delta$ is the de Rham differential on the phase space.
Pick a space-like closed oriented contour on the worldsheet, and let $\oint \vartheta$
mean the integral of $\vartheta$ over this contour.
The ``symplectic potential'' (a.k.a.\ ``$p\delta q$'') is defined as the restriction
of $\oint\vartheta$ to on-shell variations:
\begin{equation*}
 \oint\vartheta|_{\text{on-shell}} = p\delta q.
 \end{equation*}
The density of symplectic form is the exterior derivative of $\vartheta$ on the phase space:
\begin{equation*}
 \omega = \delta\vartheta.
 \end{equation*}

\subsubsection{Equivalence class of one-forms on the target space associated to a symmetry}

{\bf Definition of $\boldsymbol{\alpha\langle\xi\rangle}$.} 
Consider a symmetry $\xi$ of the target space.
If $\vartheta$ is invariant under the symmetry, then the Noether charge is the contraction of the
infinitesimal symmetry vector with~$p\delta q$, i.e., $\mathfrak{Q}\langle\xi\rangle = p\delta_{\xi}q$.
In this case, the corresponding current is
\begin{equation*}
 j\langle\xi\rangle = \iota_{\xi}\vartheta = \vartheta|_{\delta q \mapsto \delta_{\xi} q}.
 \end{equation*}
This formula for the current is only valid if the Lagrangian is invariant under the symmetry. However, the Lagrangian might change
by a total derivative
\begin{equation*}
 \delta_{\xi} L = {\rm d} i^*\alpha\langle\xi\rangle,
\end{equation*}
where $\alpha$ is some 1-form on the target space,
$i\colon \Sigma \longrightarrow X$
is the field configuration of the sigma-model, $i^*\alpha$ is the pullback of $\alpha$ to the worldsheet.
The one-form $\alpha$ depends linearly on $\xi\in\mathfrak{g}$ as a parameter.
Notice that we use angular brackets $\alpha\langle\xi\rangle$ to emphasize that the
dependence is linear.

In this case,
\begin{equation*}
 \delta_{\xi} \vartheta = \delta i^*\alpha\langle\xi\rangle
 \end{equation*}
and the current is
\begin{equation}
 j\langle\xi\rangle = \iota_{\xi}\vartheta + i^*\alpha\langle\xi\rangle.
 \label{NoetherCurrent}
\end{equation}
The Noether charge is the integral of the current over a space-like contour:
\begin{equation*}
 \mathfrak{Q}\langle{\xi}\rangle = \oint_{\tau = {\rm const}} j\langle{\xi}\rangle.
 \end{equation*}
Indeed,
\begin{equation*}
 \delta \mathfrak{Q}\langle\xi\rangle = \iota_{\xi}\oint\omega.
 \end{equation*}
In terms of the momentum variables $p_I$,
\begin{equation*}
 j\langle \xi\rangle = \bigl(\delta_{\xi}x^I\bigr) p_I + i^*\alpha\langle\xi\rangle.
 \end{equation*}
(The expansion of $j\langle\xi\rangle$ in powers of $p$ terminates at the linear term,
 because we are assuming that the symmetry is projectable to the target space.)

{\bf Ambiguities in the definition of $\boldsymbol{\alpha\langle\xi\rangle}$.} The 1-form $\alpha$ is defined up to a total derivative and up to a symmetry variation
of a one-form
\begin{equation*}
 \alpha\langle\xi\rangle \simeq \alpha\langle\xi\rangle + {\rm d} (\text{smth}\langle\xi\rangle) + {\cal L}_{\xi}(\text{smth}).
 \end{equation*}
The ambiguity of adding ${\rm d}(\text{smth}\langle\xi\rangle)$ is because the \emph{current} $j\langle\xi\rangle$
is defined only up to a total derivative.

The ambiguity of adding ${\cal L}_{\xi}(\text{smth})$ is because the \emph{Lagrangian} $L$ is given only
up to a total derivative. Replacing $L$ with $L+{\rm d}\Lambda$ with $\Lambda\in\Omega^1(X)$
is equivalent
to a canonical transformation:%
\begin{equation}
 \bigl(p_I,x^I\bigr)\mapsto \biggl(p_I + \frac{\partial}{\partial x^I} F,x^I\biggr)
 \end{equation}
for any $F$ of the form
\begin{equation*}
 F = \oint_{\tau={\rm const}} i^*\Lambda,
 \end{equation*}
where $\Lambda$ is any 1-form on the target space. This would change
\begin{equation*}
 \alpha\langle\xi\rangle \mapsto \alpha\langle\xi\rangle + {\cal L}_{\xi} \Lambda.
 \end{equation*}

\subsection{Poisson bracket of charge with current}

The charges satisfy
\begin{equation*}
 \delta_{\xi} \mathfrak{Q}\langle\eta\rangle = \{\mathfrak{Q}\langle\xi\rangle,\mathfrak{Q}\langle\eta\rangle\} = \mathfrak{Q}\langle [\xi,\eta]\rangle.
 \end{equation*}
But the Poisson bracket of charge with current is more complicated. Let us choose, for each~${\xi\in\mathfrak{g}}$,
the corresponding $\alpha$ (i.e., ``fix by hand'' the ambiguity of adding ${\rm d}(\ldots)$ to $\alpha$),
so that we have a~linear map
\begin{align*}
\mathfrak{g}\longrightarrow \Omega^1(X),
 \qquad
\xi \mapsto \alpha\langle\xi\rangle.
 \end{align*}
Then the charge-current Poisson bracket can be expressed through the Courant bracket:
\begin{gather*} \{\mathfrak{Q}\langle\xi\rangle,j\langle\eta\rangle\}=
 j\langle [(\xi,\alpha\langle\xi\rangle),(\eta,\alpha\langle\eta\rangle)]_{\tt Courant}\rangle,
\end{gather*}
where
\begin{gather}
 [
 (\xi,\alpha\langle\xi\rangle),
 (\eta,\alpha\langle\eta\rangle)
 ]_{\tt Courant} =
 (
 [\xi,\eta] ,
 {\cal L}_{v\langle\xi\rangle}\alpha\langle\eta\rangle -
 {\cal L}_{v\langle\eta\rangle}\alpha\langle\xi\rangle +
 {\rm d}\iota_{v\langle\eta\rangle}\alpha\langle\xi\rangle
 ).
 \label{CourantBracket}
 \end{gather}
Here the term ${\cal L}_{v\langle\xi\rangle}\alpha\langle\eta\rangle$ comes from
$\bigl\{\oint \xi^I p_I,\alpha\langle\eta\rangle\bigr\}$
and the term
$- {\cal L}_{v\langle\eta\rangle}\alpha\langle\xi\rangle + {\rm d}\iota_{v\langle\eta\rangle}\alpha\langle\xi\rangle = - \iota_{v\langle\eta\rangle}{\rm d}\alpha\langle\xi\rangle$
comes from
$\bigl\{\eta^I p_I,\oint\alpha\langle\xi\rangle\bigr\}$.

\subsection{Descent}

The Lie algebra differential is defined as follows \cite{FeiginFuchs,GelfandManin,Knapp}:
\begin{equation*}
 ({\rm d}_{\rm Lie}\alpha)\langle\xi\wedge\eta\rangle = {\cal L}_{v\langle\xi\rangle}\alpha\langle\eta\rangle - {\cal L}_{v\langle\eta\rangle}\alpha\langle\xi\rangle - \alpha\langle[\xi,\eta]\rangle.
 \end{equation*}
In particular, the deviation of the current from transforming covariantly is
\begin{equation}
 \delta_{\xi} j\langle\eta\rangle -
 j\langle[\xi,\eta]\rangle
 =
{\rm d}_{\rm Lie}\alpha\langle\xi\wedge\eta\rangle + {\rm d}\iota_{v\langle\eta\rangle}\alpha\langle\xi\rangle.
 \label{DeviationFromCovariance}\end{equation}

\begin{Lemma}
Exists $V\langle \xi\wedge\eta\rangle$ satisfying
\begin{equation}
 ({\rm d}_{\rm Lie}\alpha) \langle \xi\wedge\eta\rangle = {\rm d}(V\langle\xi\wedge\eta\rangle).
 \label{DescentFirstStep}\end{equation}
\end{Lemma}
\begin{proof}
The Poisson bracket of charges is
\begin{equation*}
 \{\mathfrak{Q}\langle\xi\rangle,\mathfrak{Q}\langle\eta\rangle\} = \mathfrak{Q}\langle[\xi,\eta]\rangle.
 \end{equation*}
Consider field configurations which have $p=0$ at $\tau=0$.
For such configurations, equation~(\ref{CourantBracket}) implies
\begin{equation*}
 \oint i^*
 \bigl({\cal L}_{v\langle\xi\rangle}\alpha\langle\eta\rangle - {\cal L}_{v\langle\eta\rangle}\alpha\langle\xi\rangle - \alpha\langle[\xi,\eta]\rangle\bigr)
 = 0,
 \end{equation*}
where the integral is over the spacial slice $\tau=0$.
The image of this spacial slice in the target space can be any one-dimensional closed contour.
Therefore, ${\cal L}_{v\langle\xi\rangle}\alpha\langle\eta\rangle - {\cal L}_{v\langle\eta\rangle}\alpha\langle\xi\rangle - \alpha\langle[\xi,\eta]\rangle$
is ${\rm d}$-exact, equation~(\ref{DescentFirstStep}).
\end{proof}

Consider
\begin{equation}
 W = {\rm d}_{\rm Lie}V.
 \label{WExactInCInfty}\end{equation}
We observe that $W$ is a function (0-form) and ${\rm d}W=0$;
therefore $W = \rm const$.
Equation~(\ref{WExactInCInfty}) implies that $W$ is exact as a 3-cocycle with coefficients in $C^{\infty}(X)$.
But since $W={\rm const}$, we can as well consider $W$ as a 3-cocycle of $\mathfrak g$ with coefficients in $\mathbb{R}\subset C^{\infty}(X)$.
Then we can ask if it is exact in the cochain complex with constant coefficients.
The question is: is it possible to represent $W$ as ${\rm d}_{\rm Lie}(\ldots)$ where $\ldots$ is constant on $X$?
If it is not possible, then $W$ defines a~nontrivial cohomology class with constant coefficients:
\begin{equation*}
 [W]\in H^3(\mathfrak{g},\mathbb{R}).
 \end{equation*}

\subsubsection[Case W=0]{Case $\boldsymbol{[W]=0}$}

 In this case we can assume that ${\rm d}_{\rm Lie} V = 0$, by adding a constant to $V$.

{\bf Case $\boldsymbol{H^2(\mathfrak{g},C^{\infty}(X)) = 0}$.} In this case $V = {\rm d}_{\rm Lie} u$, we can modify $\alpha \mapsto \alpha + {\rm d}u$, and then ${\rm d}_{\rm Lie}\alpha = 0$.
 Suppose that also $H^1\bigl(\mathfrak{g},\Omega^1(X)\bigr)=0$, then we can choose the Lagrangian
 to be $\mathfrak{g}$-invariant, and the currents $\mathfrak{g}$-covariant.

 In Section \ref{sec:HQinOneForms}, we consider an example where
 $H^2(\mathfrak{g},C^{\infty}(X)) = 0$ but $H^1\bigl(\mathfrak{g},\Omega^1(X)\bigr)\neq 0$. Then, the Lagrangian can not be chosen to be invariant,
 and the currents cannot be chosen in such a way that they would transform covariantly.\footnote{We want to thank the referee for correcting an error in the first version of this paper.}

{\bf Case $\boldsymbol{H^2(\mathfrak{g},C^{\infty}(X)) \neq 0}$.} In this case, $V\in H^2(\mathfrak{g},C^{\infty}(X))$ is an obstacle to finding an invariant Lagrangian.
 We may call it ``generalized universal vertex operator'', because vertex operators of
 pure spinor formalism provide an important example of this situation.
 The BRST transformations generate a supercommutative
 Lie superalgebra $\mathbb{R}^{0|2}$, and
 $H^2\bigl(\mathbb{R}^{0|2},C^{\infty}(X)\bigr) \neq 0$.
 This is precisely the BRST cohomology at ghost number two, corresponding
 to physical states, therefore $V$ is one of the physical states. This is
 the ``universal vertex operator'', see
 Section~\ref{UniversalVertex}.

\subsection{Ascent}

Equation~(\ref{DescentFirstStep}) implies
\begin{equation}
 ({\rm d}_{\rm Lie} {\rm d}\alpha)\langle\xi\wedge\eta\rangle = 0.
 \label{AscentFirstStep0}\end{equation}
Suppose that the first cohomology group of $H^1\bigl(\mathfrak{g}, \Omega^2(X)\bigr)=0$. Then
equation~(\ref{AscentFirstStep0}) implies the existence of a two-form $B$ such that
\begin{equation*}
{\rm d}\alpha\langle\xi\rangle = {\cal L}_{v\langle\xi\rangle} B.
 \end{equation*}

We denote
\begin{equation*}
 H = {\rm d} B.
 \end{equation*}
Then
\begin{equation*}
 H + W = ({\rm d}_{\rm Lie} + {\rm d})(B - \alpha + V).
 \end{equation*}
It \emph{is} often true that $H^1\bigl(\mathfrak{g}, \Omega^2(X)\bigr)=0$, by some version of Shapiro's lemma~\cite{FeiginFuchs,Knapp}. For example, this is true for the super-Poincar\'e symmetries in the
pure spinor formalism in flat space, see
Section~\ref{PureSpinorFlat}.

\section{BRST language for Lie algebra cohomology}\label{BRSTLanguage}

 Until this point we assumed that $\mathfrak{g}$ is usual (not super) Lie algebra.
 When considering the cohomology complex of Lie superalgebras, sign rules may be confusing.
 Here we will suggest a~formalism which automatically takes into account the $\pm$ signs.

Let $M$ be a $\mathfrak g$-module. Let $\rho$ denote the representation of $\mathfrak g$ in $M$:
\begin{equation*}
 \rho\colon \ {\mathfrak g}\rightarrow \mbox{End}(M).
 \end{equation*}
The odd tangent space $\Pi TG$ is defined so that for all supermanifolds $X$:
\begin{equation*}
 \mbox{Map}(X,\Pi TG) = \mbox{Map}\bigl(X\times \mathbb{R}^{0|1},G\bigr).
 \end{equation*}
Notice that the body of $\Pi TG$ is the space of ``odd loops'':
\begin{equation*}
 (\Pi TG)_{\rm rd} = \mbox{Map}\bigl(\mathbb{R}^{0|1}, G\bigr).
 \end{equation*}
As a slight variation of the construction in \cite{Alexandrov:1995kv}, we consider the following linear space:
\begin{equation*}
 C(\mathfrak{g},M)
 =
 M\otimes_G C^{\infty}(\Pi TG).
 \end{equation*}
It is possible to consider a family, parameterized by a supermanifold $X$:
\begin{equation*}
 C_X(\mathfrak{g},M)
 =
 M\otimes_G C^{\infty}(X\times\Pi TG).
 \end{equation*}
We work in a formal neighborhood of the unit of $G$, and $M$ does not have to be an integrable
representation.
Elements of $C(\mathfrak{g},M)$ are $M$-valued functions of the form
\begin{align} (x,g)\mapsto
 \rho(g(\psi))v\bigl(x,g(\psi)^{-1}\partial_{\psi}g(\psi)\bigr).
 \label{CochainAsAFunction} \end{align}
Here $v\in C^{\infty}(X\times \Pi\mathfrak{g})$,
$\psi\in C^{\infty}_{\bar{1}}(X)$ is some odd function on $X$, and we use an abbreviated notation:
\begin{equation*}
 g(\psi)^{-1}\partial_{\psi}g(\psi) = \frac{\partial}{\partial\zeta} g(\psi)^{-1}g(\psi + \zeta),
 \end{equation*}
where $\zeta$ is a Grassmann odd variable.\footnote{All we need from $X$ is a Grassmann odd ``constant'' $\psi$.
 In this approach to supermanifolds, we get it by considering families of constructions parameterized by a supermanifold $X$.}
It is useful to observe
\begin{equation*}
 \partial_{\zeta}(g(\zeta)^{-1}\partial_{\zeta}g(\zeta)) =
 -\frac{1}{2}\bigl\{g(\zeta)^{-1}\partial_{\zeta}g(\zeta),g(\zeta)^{-1}\partial_{\zeta}g(\zeta)\bigr\}.
 \end{equation*}
We may decompose into the basis of $\mathfrak g$:
\begin{equation*}
 g(\zeta)^{-1}\partial_{\zeta}g(\zeta) = \bigl(g(\zeta)^{-1}\partial_{\zeta}g(\zeta)\bigr)^A t_A.
 \end{equation*}
and introduce the ``ghost'' notation:
\begin{equation*}
 C^A = \bigl(g(\zeta)^{-1}\partial_{\zeta}g(\zeta)\bigr)^A.
 \end{equation*}
Elements of $C(\mathfrak{g},M)$ are in one-to-one correspondence
with $M$-valued polynomial functions on~$\Pi\mathfrak{g}$:
\begin{equation*}
 v\in M\otimes S^{\bullet}(\Pi\mathfrak{g})^*.
 \end{equation*}
(Here $S^{\bullet}(\ldots)$ denotes the space of symmetric tensors of any rank,
 and therefore $S^{\bullet}(\Pi\mathfrak{g})^*$ the space of polynomial functions on $\Pi\mathfrak{g}$.)
Given $v$, the corresponding function of the form equation~(\ref{CochainAsAFunction}) will be called ${\cal F}v$:
\begin{equation*}
 ({\cal F}v)(\zeta,g) = \rho(g(\zeta))v\bigl(g(\zeta)^{-1}\partial_{\zeta}g(\zeta)\bigr).
 \end{equation*}
The differential of the Lie algebra cohomology complex is denoted ${\rm d}_{\rm Lie}$ and defined as follows:
\begin{equation*}
 \partial_{\zeta}{\cal F} = {\cal F}{\rm d}_{\rm Lie}.
 \end{equation*}
These definitions work for $\mathfrak{g}$ a Lie superalgebra. If $\mathfrak{g}$ is a usual (not super) Lie algebra,
then we can identify
\begin{equation*}
 M\otimes S^{\bullet}(\Pi\mathfrak{g})^* = \operatorname{Hom}\biggl(\bigoplus_{n\geq 0}\Lambda^n \mathfrak{g}, M\biggr)
 \end{equation*}
and our definition returns to the usual definition of Lie algebra cohomology complex.

\section{Target space bicomplex}

The central role in our considerations is played by two anticommuting
 differentials. One is~${\rm d}_{\rm Lie}$~-- the differential in the Serre--Hochschild
 complex of a Lie superalgebra. Another is ${\rm d}$ -- the de~Rham differential on~$X$.
 Here we will discuss some basic properties of the bicomplex ${\rm d} + {\rm d}_{\rm Lie}$.

Consider the nilpotent operator ${\rm d}+{\rm d}_{\rm Lie}$ acting on
the cochains of~$\mathfrak{g}$ with values in differential forms on~$X$.
It is related to the de Rham operator ${\rm d}$ by a similarity transformation:
\begin{align} &{\rm d} + {\rm d}_{\rm Lie} =
 {\rm e}^{C^A\iota_{v_A}}\bigl({\rm d} + {\rm d}_{\rm Lie}^{(0)}\bigr){\rm e}^{-C^A\iota_{v_A}}, \qquad
{\rm d}_{\rm Lie}^{(0)} = \frac{1}{2}C^AC^Bf_{AB}{}^C\frac{\partial}{\partial C^C}. \label{Kalkman}
 \end{align}
where $v_A$ is the vector field on the target space corresponding to the
infinitesimal symmetry. In this work, we only need differential forms which are polynomial in the differentials.
We do not need pseudodifferential forms.
Let us assume that the de Rham cohomology of $X$ is $\mathbb{R}$ (constant functions).
Then, equation~(\ref{Kalkman}) implies that the cohomology of ${\rm d}+{\rm d}_{\rm Lie}$ is $H^{\bullet}(\mathfrak{g},\mathbb{R})$.

Let $\omega\in \Omega^n(X)$. Then
\begin{equation}
{\rm d}_{\rm Lie}\bigl(\iota^n_{v\langle C\rangle}\omega\bigr) = \frac{1}{n+1}\iota^{n+1}_{v\langle C\rangle} {\rm d}\omega.
 \label{DLieAndIota}\end{equation}
To prove this, let us rewrite equation~(\ref{Kalkman}) as follows:
\begin{equation*}
{\rm d}_{\rm Lie}\, {\rm e}^{C^A\iota_{v_A}}
 =
 \bigl[ {\rm e}^{C^A\iota_{v_A}} , {\rm d} \bigr] + {\rm e}^{C^A\iota_{v_A}} {\rm d}_{\rm Lie}^{(0)}.
 \end{equation*}
Acting by the left-hand side and the right-hand side on $\omega$, we get
\begin{equation*}
{\rm d}_{\rm Lie}\, {\rm e}^{C^A\iota_{v_A}}\omega
 =
 \bigl[ {\rm e}^{C^A\iota_{v_A}} , {\rm d} \bigr]\omega
 \end{equation*}
since ${\rm d}_{\rm Lie}^{(0)}\omega = 0$. Consider the coefficient of $C^{n+1}$. Since $\omega$ is an $n$-form,
only ${\rm e}^{C^A\iota_{v_A}} {\rm d}\omega$ contributes from $\bigl[ {\rm e}^{C^A\iota_{v_A}} , {\rm d} \bigr]\omega$ on the right-hand side, and equation~(\ref{DLieAndIota})
follows.

\section{Transformation of currents}\label{TransformationOfCurrents}

We will now investigate under which conditions the currents transform covariantly, i.e.,
\begin{equation}
 \delta_{\xi}j\langle\eta\rangle \stackrel{?}{=} j\langle[\xi,\eta]\rangle.
 \label{CurrentsTransformCovariantly}\end{equation}

\subsection{Currents transform in an extension of the adjoint representation}

Currents transform covariantly up to a total derivative:
\begin{equation}
 \delta_{\xi} j\langle\eta\rangle = j\langle[\xi,\eta]\rangle + {\rm d}(\Psi\langle\xi\rangle\langle\eta\rangle).
 \label{EntersPsi}\end{equation}
This is the definition of $\Psi\langle\xi\rangle\langle\eta\rangle$.
This defines an extension $A$ of the adjoint representation of $\mathfrak{g}$:
\begin{equation*}
 0 \longrightarrow {C^{\infty}(X)\over \mathbb{R}} \longrightarrow A \longrightarrow \adg \longrightarrow 0.
\end{equation*}
The explicit expression for the cocycle $\Psi$ follows from equations~(\ref{DeviationFromCovariance}) and (\ref{DescentFirstStep}):
\begin{gather} \Psi \mod \text{const}\in
 \mbox{Ext}^1\biggl(\mathfrak{g} ,{C^{\infty}(X)\over \mathbb{R}}\biggr)
 = H^1\biggl(\mathfrak{g} , \operatorname{Hom}\biggl(\adg ,{C^{\infty}(X)\over \mathbb{R}}\biggr)\biggr), \nonumber
 \\ \Psi\langle\xi\rangle\langle\eta\rangle=
 V\langle\xi\wedge\eta\rangle + \iota_{v\langle\eta\rangle}\alpha\langle\xi\rangle,
 \label{CocycleExt} \\ V\in
 H^2\biggl(\mathfrak{g} ,{C^{\infty}(X)\over \mathbb{R}}\biggr), \nonumber
 \\ {\rm d}V=
 {\rm d}_{\rm Lie}\alpha.
 \label{AlphaVsV} \end{gather}
We will now develop a general formalism to understand equation~(\ref{CocycleExt}).

\subsection{Some maps between complexes}

For a Lie superalgebra $\mathfrak{g}$, let $C\mathfrak{g}$ be the Lie superalgebra
generated by Lie derivatives $\cal L_{\xi}$ and contractions $\iota_{\xi}$ for $\xi\in\mathfrak{g}$,
see \cite[Section 10]{Cordes:1994fc} where it is called $\mathfrak{g}_{\rm super}$ or $\mathfrak{g}[\epsilon]$;
it has a differential.

Let $M$ be a differential $\mathfrak{g}$-module, i.e., a representation of $C\mathfrak{g}$ with a compatible
differential~${\rm d}_M$. In our case, $M$ is the space of differential forms on~$X$.

Consider a map
\begin{align*}
 \mathbb{L} \colon \ C^n(\mathfrak{g} , M) \rightarrow C^n(\mathfrak{g} , \operatorname{Hom}(\mathfrak{g} , M)),
 \qquad
\mathbb{L}(a) = \bigl[C\mapsto \bigl[\eta\mapsto {\cal L}^M_{\eta}(a(C))\bigr]\bigr].
 \end{align*}
This map is homotopy trivial:
\begin{gather} \mathbb{L} =
 \mathbb{I}{\rm d}^M_{\rm Lie} + {\rm d}^{{\rm Hom}(\mathfrak{g},M)}_{\rm Lie} \mathbb{I},
 \label{Homotopy}
\end{gather}
where
\begin{gather*}
 \mathbb{I}\colon \ C^n(\mathfrak{g} , M) \rightarrow C^{n-1}(\mathfrak{g} , \operatorname{Hom}(\mathfrak{g},M)),
 \qquad
\mathbb{I}(a) = \biggl[\eta\mapsto \eta^A{\partial\over\partial C^A} a\biggr].
 \end{gather*}
Indeed, using the notations of
Section \ref{BRSTLanguage},
\begin{equation*}
 {\cal F}\mathbb{L} = (\eta g){\partial\over\partial g} {\cal F},
 \end{equation*}
i.e., the infinitesimal left shift of $g$ by $\eta$. At the same time
\begin{equation*}
 {\cal F}\mathbb{I} = (\zeta\eta g){\partial\over\partial g} {\cal F}.
 \end{equation*}
Since ${\rm d}_{\rm Lie}$ corresponds to $\partial_{\zeta}$, equation~(\ref{Homotopy}) follows.
The construction so far only requires that~$M$ is a $\mathfrak{g}$-module.

But once $M$ is a \emph{differential} $\mathfrak{g}$-module, there is yet another map which we call $\mathbb{J}$:
\begin{align*}
 \mathbb{J} \colon \ C^n(\mathfrak{g} , M) \rightarrow C^n(\mathfrak{g} , \operatorname{Hom}(\mathfrak{g},M)),
 \qquad
 \mathbb{J}(a) = \bigl[\eta\mapsto - \iota^M_{\eta}a\bigr].
 \end{align*}
It intertwines the two complexes:
\begin{equation*}
 \mathbb{J} {\rm d}^M_{\rm Lie} = {\rm d}^{{\rm Hom}(\mathfrak{g},M)}_{\rm Lie}\mathbb{J}.
 \end{equation*}
Indeed, in notations of Section \ref{BRSTLanguage},
\begin{equation*}
 {\cal F}\mathbb{J} = - \iota^M_{\eta} {\cal F}.
 \end{equation*}
It acts just on $M$, and does not know anything about $g(\zeta)$.\footnote{For example,
\begin{gather*} \mathbb{J} {\rm d}_{\rm Lie} {\rm d}X^m =
 \mathbb{J} \biggl( -\frac{1}{2}\bigl(\gamma_L\Gamma^m{\rm d}\theta_L\bigr) -\frac{1}{2}\bigl(\gamma_R\Gamma^m{\rm d}\theta_R\bigr) \biggr) = -\frac{1}{2}\bigl(\gamma_L\Gamma^m\eta_L\bigr) -\frac{1}{2}\bigl(\gamma_R\Gamma^m\eta_R\bigr),
 \\ {\rm d}_{\rm Lie}\mathbb{J} {\rm d}X^m =
 {\rm d}_{\rm Lie}\biggl( \eta^m -\frac{1}{2}\bigl(\eta_L\Gamma^m\theta_L\bigr) -\frac{1}{2}\bigl(\eta_R\Gamma^m\theta_R\bigr) \biggr)
 \\ \hphantom{{\rm d}_{\rm Lie}\mathbb{J} {\rm d}X^m}{}=
 - \bigl(\eta_L\Gamma^m\gamma_L\bigr) - \bigl(\eta_R\Gamma^m\gamma_R\bigr) +\frac{1}{2}\bigl(\eta_L\Gamma^m\gamma_L\bigr) +\frac{1}{2}\bigl(\eta_R\Gamma^m\gamma_R\bigr)
 - \frac{1}{2}\bigl(\eta_L\Gamma^m\gamma_L\bigr) -\frac{1}{2}\bigl(\eta_R\Gamma^m\gamma_R\bigr).
 \end{gather*}}
Notice that $\mathbb{I} + \mathbb{J}$ intertwines
cochain complexes with values in $M$ and $\operatorname{Hom}(\mathfrak{g},M)$
with differential ${\rm d}_M+{\rm d}_{\rm Lie}$.

With these notations equation~(\ref{CocycleExt}) can be written as follows:
\begin{equation*}
 \Psi = \mathbb{I}V - \mathbb{J}\alpha.
 \end{equation*}
In particular, if $\alpha = 0$, then equation~(\ref{CurrentsTransformCovariantly}) is satisfied:
\begin{itemize}\itemsep=0pt
\item If Lagrangian is invariant, then currents transform covariantly.
\end{itemize}
But the precise conditions for the covariant transformation of currents are weaker than the invariance of Lagrangian,
and we will now derive it. We will start by establishing a relation to the formalism of equivariant cohomology.

\subsection{Relation to equivariant cohomology}

Various relations between $V$, $\alpha$, $W$ and $\Psi$ are encoded in the
complex known as ``BRST model of equivariant cohomology''. We will now briefly review the construction of this
complex, and explain how it is useful in our context.

We will start with considering another complex, known as ``Weil model''. This is a direct product
of two complexes, and the differential is the sum of two differentials, equation~(\ref{DifferentialInWeilModel}).
The BRST model is related to it by a similarity transformation, equation~(\ref{KalkmanMap}).

\subsubsection{Weil model}

Weil algebra $\cal W$ is a differential graded super-commutative algebra formed
by generators\footnote{Often $C$ is called $A$ and $\eta$ is called $F$;
 in \cite{Cordes:1994fc} $C$ is $\theta$ and $\eta$ is $\phi$.}~$C$ and~$\eta$
with the differential
\begin{equation*}
{\rm d}_{\cal W} = \eta {\partial\over\partial C} + \frac{1}{2}[C,C]{\partial\over\partial C} + [C,\eta] {\partial\over\partial\eta}.
 \end{equation*}
We consider the space of functions of $C$ and $\eta$ with values in $M$, with the differential
\begin{equation}
{\rm d}_{\rm tot} = {\rm d}_{\cal W} + {\rm d}_M,
 \label{DifferentialInWeilModel}\end{equation}
where ${\rm d}_M$ is the differential of~$M$.

\subsubsection{BRST model}

We have
\begin{equation}
 {\rm e}^{-\iota_A C^A} ({\rm d}_{\cal W} + {\rm d}_M) {\rm e}^{\iota_A C^A} =
{\rm d}_{\rm Lie} +{\rm d}_M + \mathbb{I} + \mathbb{J},
 \label{KalkmanMap}\end{equation}
where ${\rm d}_{\rm Lie}$ is the Lie algebra cohomology differential with coefficients
in 
$\operatorname{Hom}(S^{\bullet}\adg, M)$.
\big(In this language $\mathbb{I} = \eta{\partial\over\partial C}$.\big)
The nilpotence of ${\rm d}_{\rm Lie} + {\rm d}_M + \mathbb{I} + \mathbb{J}$ follows from the nilpotence of~${{\rm d}_{\cal W} + {\rm d}_M}$.

We have
\begin{equation}
 ({\rm d}_{\rm Lie} + {\rm d}_M + \mathbb{I} + \mathbb{J})(V - \alpha) = W + \Psi - {\rm d}_M\alpha - \mathbb{I}\alpha,
 \label{DtotExactExpression}\end{equation}
where
\begin{gather*} W =
{\rm d}_{\rm Lie}V, \qquad \Psi =
 \mathbb{I}V - \mathbb{J}\alpha.
 \end{gather*}

It is useful to consider the ``total ghost number'' defined as follows:
\begin{gather*}
 N_{\rm tot} = \mbox{(form rank)} + C{\partial\over\partial C} + 2 \eta{\partial\over\partial \eta},
\\
N_{\rm tot}(V-\alpha) = 2(V-\alpha),
 \\
N_{\rm tot}(W + \Psi - {\rm d}\alpha - \mathbb{I}\alpha) = 3(W + \Psi - {\rm d}\alpha - \mathbb{I}\alpha).
 \end{gather*}
The grading table is
\begin{center}
\begin{tabular}{|l|l|l|l|l|}
\hline \rule{0pt}{12pt} &$N_{\rm tot}$&form&$C{\partial\over\partial C}$&$2 \eta{\partial\over\partial \eta}$\\
\hline $V$&$2$&$0$&$2$&$0$\\
\hline $\alpha$&$2$&$1$&$1$&$0$\\
\hline $W$&$3$&$0$&$3$&$0$\\
\hline $\Psi$&$3$&$0$&$1$&$2$\\
\hline ${\rm d}\alpha$&$3$&$2$&$1$&$0$\\
\hline $\mathbb{I}\alpha$&$3$&$1$&$0$&$2$\\
\hline \end{tabular}
\end{center}

\noindent
The term with
$\bigl({\rm form},C{\partial\over\partial C},2 \eta{\partial\over\partial \eta}\bigr) = (1,2,0)$
is missing in equation~(\ref{DtotExactExpression}), because of equation~(\ref{AlphaVsV}).
The nilpotence of ${\rm d}_{\cal W} + {\rm d}_M$ implies the nilpotence of ${\rm d}_{\rm Lie} + {\rm d}_M + \mathbb{I} + \mathbb{J}$;
this is called ``BRST model'' in~\cite{Cordes:1994fc}.

Therefore,
\begin{align}
{\rm d}_{\rm Lie} \Psi + \mathbb{I} W = 0
 \label{DLiePsi}, \qquad
{\rm d}_M\Psi = \mathbb{J}{\rm d}\alpha + {\rm d}_{\rm Lie} \mathbb{I}\alpha.
 \end{align}
Since $W$ is a constant
\begin{equation*}
{\rm d}_{\rm Lie} \Psi = 0 \mod \text{const}.
 \end{equation*}
(By $\mbox{const}$ we mean constant in the target space, i.e., independent of $X$, $\theta$, $\lambda$;
 it only depends on~$C$ and~$\eta$.)
Indeed, this is necessary for the consistency with equation~(\ref{EntersPsi}).

Equation~(\ref{DtotExactExpression}) is not completely general. Without changing the range of gradings on the right-hand side,
we can generalize it as follows:
\begin{align*} ({\rm d}_{\rm Lie} + {\rm d}_M + \mathbb{I} + \mathbb{J})(V - \alpha + \varphi) =
 W + (\Psi+{\rm d}_{\rm Lie}\varphi) - {\rm d}_M\alpha - (\mathbb{I}\alpha - {\rm d}_M\varphi),
\end{align*}
where
\begin{gather*}
\varphi \in C^0\bigl(\mathfrak{g},\operatorname{Hom}\bigl(\adg,M^0\bigr)\bigr) = \operatorname{Hom}\bigl(\adg,M^0\bigr).
\end{gather*}
The gradings on the right-hand side are still $(0,3,0)$, $(0,1,2)$ and $(1,0,2)$.
But $\varphi$ can be removed by adding
an exact expression. Indeed, given $\psi\in C^1\bigl(\mathfrak{g},M^0\bigr)$, the exact expression is: $V' - \alpha' + \varphi' = ({\rm d}_{\rm Lie} + {\rm d}_M + \mathbb{I} -\mathbb{J})\psi$;
$\alpha'=-{\rm d}_M\psi$, $V' = {\rm d}_{\rm Lie}\psi$ and $\varphi' = \mathbb{I}\psi$.
This corresponds to the ambiguity of adding a total derivative to~$\alpha$.

\subsection{Obstacles to covariance of currents}

Equation~(\ref{DLiePsi}) implies that the first obstacle to the covariance of currents is
\begin{equation}
 [\mathbb{I}W]\in H^2(\mathfrak{g},\operatorname{Hom}(\adg, \mathbb{R})).
\label{IWObstacle}\end{equation}
If $[\mathbb{I}W]=0$, then we correct $\Psi$ by adding to it a constant, so that
${\rm d}_{\rm Lie}\Psi = 0$.
In that case, the next obstacle is
\begin{equation}
 [\Psi]\in H^1(\mathfrak{g},\operatorname{Hom}(\adg, C^{\infty}(X))).
 \label{PsiObstacle}\end{equation}
In many situations, this cohomology group is zero due to Shapiro's lemma, and currents transform covariantly.

To summarize, the condition of the covariance of the currents is that we can choose $\Psi=\mbox{const}$.
When this is so, equation~(\ref{CocycleExt}) implies
\begin{equation}
 \iota_{v\langle\eta\rangle}\alpha\langle\xi\rangle + (\xi\leftrightarrow\eta) = \mbox{const}.
 \label{Symmetrization}\end{equation}
This follows from equation~(\ref{CourantBracket}) and equation~(\ref{CurrentsTransformCovariantly}) symmetrized under $\xi\leftrightarrow\eta$.
When equation~(\ref{Symmetrization}) holds, equation~(\ref{CurrentsTransformCovariantly}) becomes equivalent to
\begin{equation*}
 {\cal L}_{v\langle\xi\rangle}\alpha\langle\eta\rangle -
 {\cal L}_{v\langle\eta\rangle}\alpha\langle\xi\rangle -
 \alpha\langle[\xi,\eta]\rangle
 =
{\rm d}\biggl(\frac{1}{2}(\iota_{v\langle\xi\rangle}\alpha\langle\eta\rangle - \iota_{v\langle\eta\rangle}\alpha\langle\xi\rangle)\biggr).
\end{equation*}
This means that we can choose $V$ as follows:
\begin{equation*}
 V\langle\xi\wedge\eta\rangle
 =
 \frac{1}{2}\bigl(\iota_{v\langle\xi\rangle}\alpha\langle\eta\rangle - \iota_{v\langle\eta\rangle}\alpha\langle\xi\rangle\bigr).
\end{equation*}
To summarize, here are the conditions for the currents transforming covariantly,
i.e., $\delta_{\xi}j\langle\eta\rangle = j\langle[\xi,\eta]\rangle$:
\begin{align}
&\iota_{v\langle\eta\rangle}\alpha\langle\xi\rangle + (\xi\leftrightarrow\eta) =
 \mathfrak{C}\langle\xi\bullet\eta\rangle
 , \qquad \mbox{where $\mathfrak{C}\langle\xi\bullet\eta\rangle$ is constant}, \label{FirstConditionForCovariance}
 \\
 & V\langle\xi\wedge\eta\rangle=
 \frac{1}{2}\bigl(\iota_{v\langle\xi\rangle}\alpha\langle\eta\rangle - \iota_{v\langle\eta\rangle}\alpha\langle\xi\rangle\bigr),
 \label{SecondConditionForCovariance} \end{align}
where $\xi\bullet\eta$ stands for symmetric tensor product, and $V$ is from equation~(\ref{DescentFirstStep}).

\subsection{Strong covariance condition}

Equation~(\ref{Symmetrization}) can be strengthened, by requiring:
$\mathfrak{C}\langle\xi\bullet\eta\rangle = 0$. Then we have
\begin{equation*}
 \iota_{v\langle\xi\rangle}\alpha\langle\eta\rangle + \iota_{v\langle\eta\rangle}\alpha\langle\xi\rangle = 0.
 \end{equation*}
This is equivalent to the existence of a 2-form $b\in\Omega^2(X)$ such that
\begin{equation*}
 \alpha\langle\xi\rangle = \iota_{v\langle\xi\rangle}b.
\end{equation*}
Then
\begin{align}
({\rm d}_{\rm Lie}\alpha)\langle\xi\wedge\eta\rangle=
 {}& {\cal L}_{v\langle\xi\rangle}\alpha\langle\eta\rangle -
 {\cal L}_{v\langle\eta\rangle}\alpha\langle\xi\rangle -
 \alpha\langle[\xi,\eta]\rangle
 = \frac{1}{2}\bigl(
 {\cal L}_{v\langle\xi\rangle}\iota_{\eta} + \iota_{v\langle\eta\rangle}{\cal L}_{\xi}
 - (\xi\leftrightarrow\eta)\bigr)b \nonumber
 \\ =
 {}& {\rm d}\iota_{v\langle\xi\rangle}\iota_{\eta}b - \iota_{v\langle\xi\rangle}\iota_{\eta}{\rm d}b.
 \label{DLieVsB}
 \end{align}
In particular, if $\iota_{v\langle\xi\rangle}\iota_{v\langle\eta\rangle}{\rm d}b=0$,
then $V\langle\xi\wedge\eta\rangle = \iota_{v\langle\xi\rangle}\iota_{v\langle\eta\rangle}b$
(a.k.a.\ $\iota_{v\langle\xi\rangle\wedge v\langle\eta\rangle}b$)
satisfies the cocycle condition:
\begin{equation*}
 ({\rm d}_{\rm Lie}V)\langle\xi\wedge\eta\wedge\zeta\rangle = \bigl[{\rm d},\iota_{v\langle\xi\rangle\wedge v\langle\eta\rangle \wedge v\langle\zeta\rangle}\bigr]b = 0.
 \end{equation*}
Under this strong covariance condition,
\begin{equation*}
 ({\rm d}_{\rm Lie} + {\rm d} + \mathbb{I} + \mathbb{J})(V - \alpha +b) = H = {\rm d}b.
 \end{equation*}

\subsection{Asymmetric covariance condition}

It is also possible to formulate a weaker condition. We can ask that for some two
symmetries~$\xi$ and~$\eta$:
\begin{equation*}
 \delta_{\xi}j\langle\eta\rangle = j\langle[\xi,\eta]\rangle
 \end{equation*}
but perhaps $\delta_{\eta}j\langle\xi\rangle \neq j\langle[\eta,\xi]\rangle$.
This weaker condition follows directly from equations~(\ref{CourantBracket}) and~(\ref{DescentFirstStep}):
\begin{equation}
 V\langle\xi\wedge\eta\rangle = - \iota_{v\langle\eta\rangle}\alpha\langle\xi\rangle \mod\mbox{const}.
 \label{AsymmetricCovarianceCondition}
 \end{equation}
This is \emph{not} symmetric with respect to
$\xi\leftrightarrow\eta$. For $\Psi\langle\eta\rangle\langle\xi\rangle$ we have,
then
\begin{equation*}
 \Psi\langle\eta\rangle\langle\xi\rangle = \iota_{v\langle\eta\rangle}\alpha\langle\xi\rangle + \iota_{v\langle\xi\rangle}\alpha\langle\eta\rangle
 \mod\mbox{const}.
 \end{equation*}

\section{Case of WZW model}\label{WZW}

Here we will relate our considerations to the approach of \cite{Witten:1991mm} to WZW model.
The target space is the group manifold, and the symmetry is some combination of left and right shifts.

The non-invariance of the Lagrangian is due to the $B$-field not being strictly
invariant:
\begin{equation}
 {\cal L}_{v\langle\xi\rangle} B = \iota_{v\langle\xi\rangle} {\rm d}B + {\rm d}\iota_{v\langle\xi\rangle}B.
 \label{LB}\end{equation}
A special role in \cite{Witten:1991mm} is played by the one-forms $\lambda_a$ such that
\begin{equation}
 \iota_{v\langle\xi\rangle}{\rm d}B = \xi^a {\rm d}\lambda_a.
 \label{Lambdas}\end{equation}
Remember that $\alpha\langle\xi\rangle$ is defined up to ${\rm d}(\ldots)$.
Equations~(\ref{LB}) and (\ref{Lambdas}) imply that we can choose~$\alpha$ as follows:
\begin{equation}
 \alpha\langle\xi\rangle = \xi^a \lambda_a + \iota_{v\langle\xi\rangle} B.
 \label{DefLambdaWZW}\end{equation}
Therefore,
\begin{equation}
 \iota_{v\langle\eta\rangle}\alpha\langle\xi\rangle + (\xi\leftrightarrow\eta)
 =
 \eta^a\xi^b\bigl(\iota_{v\langle t_a\rangle}\lambda_b + \iota_{v\langle t_b\rangle}\lambda_a\bigr).
 \label{ObstacleToGauging}\end{equation}
The obstacle to gauging in \cite{Witten:1991mm} was the right-hand side of
equation~(\ref{ObstacleToGauging}), and is therefore very similar to equation~(\ref{FirstConditionForCovariance}).
All currents in WZW model do transform covariantly, therefore both
equations~(\ref{FirstConditionForCovariance}) and (\ref{SecondConditionForCovariance}) are satisfied.
However, only a subgroup can be gauged, where the constant on the right-hand side of equation~(\ref{FirstConditionForCovariance}) is zero.
The obstacle to gauging is $\mathfrak{C}\langle\xi\bullet\eta\rangle$.

It was shown in \cite{Witten:1991mm} that
\begin{equation*}
 {\cal L}_{v\langle\xi\rangle}\lambda\langle\eta\rangle = \lambda\langle[\xi,\eta]\rangle.
 \end{equation*}
Now we have
\begin{align*}
 {\cal L}_{v\langle\xi\rangle}\alpha\langle\eta\rangle -
 {\cal L}_{v\langle\eta\rangle}\alpha\langle\xi\rangle -
 \alpha\langle[\xi,\eta]\rangle =
{\rm d}\iota_{v\langle\xi\rangle}\iota_{v\langle\eta\rangle}B - \iota_{v\langle\xi\rangle}\iota_{v\langle\eta\rangle}{\rm d}B
 +\lambda\langle[\xi,\eta]\rangle.
 \end{align*}
The last two terms actually cancel:
\begin{equation*}
 \iota_{v\langle\xi\rangle}\iota_{v\langle\eta\rangle}{\rm d}B = \lambda\langle[\xi,\eta]\rangle
 \end{equation*}
and therefore
\begin{equation*}
 V\langle\xi\wedge\eta\rangle = \iota_{v\langle\xi\rangle}\iota_{v\langle\eta\rangle}B =
 \frac{1}{2}\bigl(\iota_{v\langle\xi\rangle}\alpha\langle\eta\rangle - \iota_{v\langle\eta\rangle}\alpha\langle\xi\rangle\bigr).
 \end{equation*}
In the language of ghosts,
\begin{equation*}
 V = \frac{1}{2}\iota_{v\langle C\rangle}\alpha\langle C\rangle.
 \end{equation*}

Then, equation~(\ref{DLieAndIota}) implies that
\begin{equation*}
 W\langle\xi\wedge\eta\wedge\zeta\rangle = \iota_{v\langle\xi\rangle}\iota_{v\langle\eta\rangle}\iota_{v\langle\zeta\rangle} {\rm d}B.
 \end{equation*}
We have
\begin{equation}
 ({\rm d}_{\rm Lie} + {\rm d} + \mathbb{I} + \mathbb{J})(V - \alpha + B) = W + \Psi + H + \lambda,
 \label{DtotWZW}\end{equation}
where $\lambda$ is from equation~(\ref{DefLambdaWZW}), considered as element of
$C^0\bigl(\mathfrak{g},\operatorname{Hom}\bigl(\mathfrak{g},\Omega^1(X)\bigr)\bigr)$, and other elements are
\begin{gather*} W=
{\rm d}_{\rm Lie} V \in C^3(\mathfrak{g},\mathbb{R}),
 \\ \Psi=
 \mathbb{I}V - \mathbb{J}\alpha \in C^1(\mathfrak{g},\operatorname{Hom}(\mathfrak{g},\mathbb{R})),
 \\ H=
{\rm d}B \in C^0\bigl(\mathfrak{g},\Omega^3(X)\bigr),
 \\ \lambda=
 -\mathbb{I}\alpha + \mathbb{J}B \in C^0\bigl(\mathfrak{g},\operatorname{Hom}\bigl(\mathfrak{g},\Omega^1(X)\bigr)\bigr).
 \end{gather*}
We have
\begin{align}
 & \mathbb{I}W + {\rm d}_{\rm Lie}\Psi = 0, 
 \nonumber\\
 & \mathbb{J}H + {\rm d}\lambda = 0,
 \label{IotaHisExact} \\
 & \Psi = \mathfrak{C}(\eta,C),
 \label{PsiVsC} \\
 & \mathbb{I}\Psi + \mathbb{J}\lambda = 0.
 \label{JLambdaIsIPsi} \end{align}
Notice that $H + \lambda$ is equivariantly closed (i.e., annihilated by ${\rm d} + \mathbb{J}$)
iff $\mathbb{J}\lambda = 0$. Equation~(\ref{JLambdaIsIPsi}) implies that this is equivalent to
\begin{equation*}
 \mathbb{I}\Psi = 0.
 \end{equation*}
Under this condition, also $W + \Psi$ is ${\rm d}_{\rm Lie} + \mathbb{I}$-closed. Equation~(\ref{PsiVsC}) implies that this is equivalent
to the vanishing of the obstacle to gauging in \cite{Witten:1991mm}:
\begin{equation*}
 \mathfrak{C}(\eta,\eta) = 0.
 \end{equation*}
Equation~(\ref{DtotWZW}) is the equivariant analogue of the WZW transgression relation
\begin{equation*}
 ({\rm d}_{\rm Lie} + {\rm d})(V - \alpha + B) = W + H
 \end{equation*}
replacing
\begin{gather*} {\rm d} \mapsto
 {\rm d} + \mathbb{J},
 \qquad {\rm d}_{\rm Lie} \mapsto
 {\rm d}_{\rm Lie} + \mathbb{I},
 \qquad H \mapsto
 H + \lambda,
 \qquad W \mapsto
 W + \Psi,
 \end{gather*}
${\rm d}$ with ${\rm d} + \mathbb{J}$ and ${\rm d}_{\rm Lie}$ with ${\rm d}_{\rm Lie} + \mathbb{I}$.
Notice that equation~(\ref{DtotWZW}) does not require the cancellation of anomaly in the sense of~\cite{Witten:1991mm}. In fact, $\mathbb{J}\lambda$ cancels with $\mathbb{I}\Psi$, equation~(\ref{JLambdaIsIPsi}).

\section{Pure spinor sigma-model}

In this section, we will give a very brief review of the pure spinor worldsheet sigma-model

Most of the applications in this paper are to the Type II pure spinor superstring sigma-model~\cite{Berkovits:2000fe,Berkovits:2001ue}.
This sigma-model can be constructed for any consistent Type II supergravity background.
The set of fields includes some matter fields taking values in a supermanifold (space-time),
and the so-called ``pure spinor ghosts''. The pure spinor ghosts are called~$\lambda^{\alpha}_L$ and~$\lambda^{\hat{\alpha}}_R$,
they are constrained to take values in a cone:
\begin{equation*}
 \bigl(\lambda_L^{\alpha}\Gamma_{\alpha\beta}^m\lambda_L^{\beta}\bigr)
 =
 \bigl(\lambda_R^{\hat{\alpha}}\Gamma_{\hat{\alpha}\hat{\beta}}^m\lambda_R^{\hat{\beta}}\bigr)
 =
 0,
 \end{equation*}
where $\Gamma^m_{\alpha\beta}$ are the ten-dimensional gamma-matrices.

In special backgrounds, such as flat space and ${\rm AdS}_5\times S^5$, the sigma-model has global symmetries,
basically the isometries of the background. But, crucially, there are two fermionic symmetries present
at \emph{any} background. They are denoted $Q_L$ and $Q_R$ and called ``BRST symmetries''.
They are both nilpotent and anticommuting:
\begin{equation*}
 \{Q_L,Q_L\} = \{Q_R,Q_R\} = \{Q_L,Q_R\} = 0.
 \end{equation*}
Moreover, the Noether current corresponding to $Q_L$ can be chosen to be holomorphic,
and the Noether current corresponding to $Q_R$ antiholomorphic. They are
\begin{equation*}
 j_+\langle Q_L\rangle = \lambda_L^{\alpha} {\rm d}_{\alpha +},\qquad
 j_-\langle Q_R\rangle = \lambda_R^{\hat{\alpha}} {\rm d}_{\hat{\alpha}-}.
 \end{equation*}
In flat space,
\begin{gather} {\rm d}_{\alpha +}=
 p_{\alpha +} + \frac{1}{2}\partial_+ x^m \Gamma^m_{\alpha\beta} \theta_L^{\beta} +
 {1\over 8}\bigl(\theta_L^{\gamma}\Gamma_{\gamma\delta}^m\partial_+\theta_L^{\delta}\bigr)\Gamma^m_{\alpha\beta}\theta_L^{\beta},
 \label{DPlus} \\ {\rm d}_{\hat{\alpha} -}=
 p_{\hat{\alpha} -} + \frac{1}{2}\partial_- x^m \Gamma^m_{\hat{\alpha}\hat{\beta}} \theta_R^{\hat{\beta}} +
 {1\over 8}\bigl(\theta_R^{\hat{\gamma}}\Gamma_{\hat{\gamma}\hat{\delta}}^m\partial_-\theta_R^{\hat{\delta}}\bigr) \Gamma^m_{\hat{\alpha}\hat{\beta}}\theta_R^{\hat{\beta}}.
 \label{DMinus} \end{gather}

\section{Universal vertex operator}\label{UniversalVertex}

Pure spinor superstring in any curved background has a natural deformation,
 which is the change of the overall scale of the Lagrangian.
 It can be thought of as varying the string tension~$1\over \alpha'$.
 In this section, we will explain how this ``universal vertex operator''
 is related to $\alpha\langle Q\rangle$, equation~(\ref{UniversalDescent2}).

\subsection{Definition for arbitrary curved background}

Equation~(\ref{DescentFirstStep}) implies
\begin{equation}
 {\cal L}_Q\alpha\langle Q\rangle = \frac{1}{2} {\rm d} V\langle Q\wedge Q\rangle.
 \label{LQAQ}\end{equation}
An important role is played by the \emph{ghost number symmetry}.
In other words, we assume that $X$ is a \emph{graded supermanifold} of degree one.
(Where $x$ and $\theta$ are of degree zero, and $\lambda$ of degree one.)
Let $N$ be the grading vector field
\begin{equation*}
 N = \lambda_L^{\alpha}{\partial\over\partial\lambda_L^{\alpha}} +
 \lambda_R^{\hat{\alpha}}{\partial\over\partial\lambda_R^{\hat{\alpha}}}
 \end{equation*}
and $f_0$ be a function of degree zero (i.e., a function of $x$ and $\theta$ only).
A direct examination of the sigma-model Lagrangian in \cite{Berkovits:2001ue} shows that
it is exactly invariant under $N$ (since it is just a~phase rotation of ghosts and their momenta).
Therefore,
\begin{equation*}
 \alpha\langle N\rangle = 0.
 \end{equation*}
Using the freedom of adding a total derivative to $\alpha\langle Q\rangle$, we can choose $\alpha\langle Q\rangle$ so that
\begin{equation}
 \iota_N\alpha\langle Q\rangle = 0.
 \label{IotaNAlpha}
 \end{equation}
Indeed, if $\iota_N\alpha\neq 0$, we modify
\begin{equation*}
 \alpha \mapsto \alpha - {\rm d}\iota_N\alpha.
 \end{equation*}
Also notice that we can choose $V(N,Q)=0$.
This is because $N$ is just a phase rotation, therefore ${\cal L}_N\alpha\langle Q\rangle = \alpha\langle [N,Q]\rangle$.
With equation~(\ref{IotaNAlpha}) this implies
\begin{equation*}
 \Psi\langle Q_L\rangle\langle N\rangle = \Psi\langle Q_R\rangle\langle N\rangle = 0.
 \end{equation*}
Therefore,
\begin{align}
 & Q_L j\langle N\rangle = - j\langle Q_L\rangle,
 \label{QLExact} \\
 & Q_R j\langle N\rangle = - j\langle Q_R\rangle.
 \label{QRExact} \end{align}
In other words, both $j\langle Q_L\rangle$ and $j\langle Q_R\rangle$ are BRST exact.
Taking into account that ${\cal L}_NV = 2V$, equation~(\ref{LQAQ}) implies
\begin{equation}
 V = \frac{1}{2}\iota_N {\rm d} V = \iota_N {\cal L}_Q\alpha = - \iota_Q\alpha.
 \label{UniversalVasIotaAlpha}\end{equation}
Given equation~(\ref{CocycleExt}), this is consistent with BRST currents being themselves BRST invariant,
in any background. As we know, they are both BRST exact, equations~(\ref{QLExact}) and (\ref{QRExact}).

To summarize:
\begin{align}
 QL =
{\rm d}\alpha\langle Q\rangle, 
\qquad {\cal L}_Q\alpha\langle Q\rangle =
 {\rm d}\biggl(-\frac{1}{2}\iota_Q\alpha\langle Q\rangle\biggr).
 \label{UniversalDescent2} \end{align}
We must stress that this descent procedure does not require the use of the worldsheet equations of motion.
The second step, equation~(\ref{UniversalDescent2}), involves purely geometrical objects in the target space:
vector field $Q$ and one-form $\alpha\langle Q\rangle$. This works for all pure spinor backgrounds.

\subsection{Generalization}

Suppose that $\mathfrak{g}$ is a graded Lie superalgebra. The grading operator will be denoted
$n\colon \mathfrak{g}\rightarrow\mathfrak{g}$. Let us assume that $n$ is diagonalizable, and all its eigenvalues nonzero
\big(in particular, exists $n^{-1}$\big).
Moreover, we will assume that there is a vector field $N$ on $X$, a symmetry of the Lagrangian, such that
\begin{gather*} [N,v\langle\xi\rangle]=
 v\langle n\xi\rangle,
 \qquad {\cal L}_N\alpha\langle\xi\rangle=
 \alpha\langle n\xi\rangle,
 \qquad {\cal L}_NV\langle\xi\wedge\eta\rangle=
 V\langle n\xi\wedge\eta\rangle + \langle\xi\wedge n\eta\rangle.
 \end{gather*}

We can modify $\alpha$ as follows:
\begin{equation*}
 \alpha\langle\xi\rangle
 \mapsto
 \alpha\langle\xi\rangle - {\rm d}\iota_N\alpha\bigl\langle n^{-1}\xi\bigr\rangle.
 \end{equation*}
Then the modified $\alpha$ satisfies
\begin{equation*}
 \iota_N\alpha = 0.
 \end{equation*}
We have
\begin{align}
 V\langle n\xi\wedge \eta + \xi\wedge n\eta\rangle &{}= \iota_N {\rm d}V\langle\xi\wedge\eta\rangle
= \iota_N \bigl({\cal L}_{v\langle\xi\rangle}\alpha\langle\eta\rangle -
{\cal L}_{v\langle\eta\rangle}\alpha\langle\xi\rangle -
\alpha\langle[\xi,\eta]\rangle
\bigr) \nonumber
\\&{} =
 (-)^{\bar{\xi}}\iota_{v\langle n\xi\rangle}\alpha\langle\eta\rangle - (-)^{\bar{\eta}}\iota_{v\langle n\eta\rangle}\alpha\langle\xi\rangle.
 \label{UsingN} \end{align}
Suppose that equation~(\ref{FirstConditionForCovariance}) is satisfied, then equation~(\ref{UsingN}) implies
\begin{equation*}
 V\langle \xi\wedge\eta\rangle = \frac{1}{2}\bigl(\iota_{v\langle\xi\rangle}\alpha\langle\eta\rangle - (\xi\leftrightarrow\eta)\bigr).
 \end{equation*}
This is a generalization of equation~(\ref{UniversalVasIotaAlpha}). If we assume that all
eigenvalues of $n$ are positive, then it follows that $V$ is a cocycle:
\begin{equation*}
{\rm d}_{\rm Lie}V = 0.
 \end{equation*}
If its cohomology class in $H^2(\mathfrak{g},C^{\infty}(X))$ is nonzero, then it is a generalization of the universal vertex operator.

\section{Contact terms}\label{ContactTerms}

We will here study the contact terms in the OPE of left and right BRST currents.
 These contact terms do not have classical analogue.

\subsection{BRST currents and universal vertex}

Since $j\langle Q_L\rangle$ is holomorphic and $j\langle Q_R\rangle$ is antiholomorphic, their OPE
is purely a contact term:
\begin{equation*}
 j\langle Q_L\rangle (z,\bar{z}) j\langle Q_R\rangle (0,0) = V \delta^2(z,\bar{z}).
 \end{equation*}
Since both currents are exactly BRST invariant, it follows:
\begin{equation*}
 (Q_L + Q_R)V = 0.
 \end{equation*}
We \emph{conjecture} that $V$ is in fact the universal vertex operator.
(In particular, this implies that the universal vertex operator has ghost number $(1,1)$.)

In flat space-time, the contact term is somewhat ill-defined,
since the conserved currents are
only well-defined up to terms which vanish on-shell.
However, in the general curved space-time, there is no ambiguity.
Indeed, all equations of motion have the conformal dimension $(1,1)$,
except those for pure spinor ghosts, which have dimension $(1,0)$ or $(0,1)$
(i.e., $\partial_-\lambda_L = \cdots$). The only possibility would be
adding something proportional to $\partial_+\lambda_R$ to $j\langle Q_L\rangle$,
and something proportional to $\partial_-\lambda_L$ to $j\langle Q_R\rangle$.
But that would have wrong ghost number.

We computed the contact terms in flat space (with the most economical definition of the currents)
and in ${\rm AdS}_5\times S^5$. In both cases,
the coefficient of $\delta^2(z,\bar{z})$ in the contact term is the universal vertex operator.

\subsection[Case of AdS\_5 times S\^{}5]{Case of $\boldsymbol{{\rm AdS}_5\times S^5}$}

The easiest way to compute them is to use
\begin{equation*}
 Q_L Q_R \bigl({\rm STr}\bigl(\bigl(w^L_{\alpha +} \lambda_L^{\alpha}\bigr)(z,\bar{z})\bigr)
 {\rm STr}\bigl(\bigl(w^R_- \lambda_R\bigr)(0,0)\bigr)\bigr)
 = \text{regular}.
 \end{equation*}
On the other hand, this is equal to
\begin{align*}
 & j_+\langle Q_L\rangle (z,\bar{z}) j_-\langle Q_R\rangle (0,0)
 + (Q_Rj_+\langle Q_L\rangle(z,\bar{z}))
 {\rm STr}\bigl(w^R_- \lambda_R\bigr)(0,0).
 \end{align*}
In the second term, $Q_Rj_+\langle Q_L\rangle(z,\bar{z})$ is zero on-shell. Whenever something is zero on-shell, there is
a corresponding vector field on the space of fields. In particular, the vector field corresponding
to $Q_Rj_+\langle Q_L\rangle(z\bar{z})$ is ${\rm STr}\bigl(\lambda_L {\delta\over\delta w^R_-}\bigr)$.
This implies
\begin{equation*}
 j_+\langle Q_L\rangle (z\bar{z})\; j_-\langle Q_R\rangle (0,0) = {\rm STr}(\lambda_L\lambda_R)\delta^2(z,\bar{z}).
 \end{equation*}

\subsection{Case of flat space}

The contact term is due to the contraction of $\partial_+x$ and $\partial_- x$ in equations~(\ref{DPlus}) and (\ref{DMinus}):
\begin{equation*}
 j_+\langle Q_L\rangle (z\bar{z})\; j_-\langle Q_R\rangle (0,0) = (\theta_L\Gamma^m\lambda_L)(\theta_R\Gamma^m\lambda_R)\delta^2(z,\bar{z}).
 \end{equation*}

\section{Pure spinor sigma-model in flat space}\label{PureSpinorFlat}

\subsection{Symmetries of the pure spinor sigma-model in flat space}

The (manifest) symmetries of the pure spinor sigma-model in flat space are:
\begin{itemize}\itemsep=0pt
\item The super-Poincar\'e group $\mathfrak{SP}$ which is the semi-direct sum of
 the supersymmetry group (translations and supersymmetries) and the Lorentz group.
 The Lie superalgebra of the super-Poincar\'e group will be denoted $\mathfrak{sP}$
 (notation: $\mathfrak{sP}=$ $\mathfrak{s}$uper-$\mathfrak{P}$oincaré)
\item Two odd nilpotent symmetries $Q_L$ and $Q_R$. The corresponding ghosts will be called $\chi_L$ and $\chi_R$.
\end{itemize}

The Lie superalgebra cohomology differential is
\begin{equation*}
{\rm d}_{\mathfrak{sP}\oplus{\bf R}^{0|2}_{\rm BRST}} = {\rm d}_\mathfrak{sP} + \chi_L Q_L + \chi_R Q_R.
 \end{equation*}
Explicitly, each part is
\begin{gather*}
{\rm d}_{\mathfrak{sP}} =
 c^\mu\frac{\partial}{\partial X^\mu} +
 \gamma_L^{\alpha}\biggl(\frac{\partial}{\partial\theta^{\alpha}_L} - \frac{1}{2}\theta^{\beta}_L\Gamma^\mu_{\beta\alpha}\frac{\partial}{\partial X^\mu}\biggr) + \gamma_R^{\hat\alpha}\biggl(\frac{\partial}{\partial\theta^{\hat\alpha}_R} - \frac{1}{2}\theta^{\hat\beta}_R\Gamma^\mu_{\hat\beta\hat\alpha}\frac{\partial}{\partial X^\mu}\biggr)
 \\ \hphantom{{\rm d}_{\mathfrak{sP}} =}{}
+ \frac{1}{2}\gamma_L^\alpha \gamma_L^\beta \Gamma^\mu_{\alpha\beta}\frac{\partial}{\partial c^\mu}
+ \frac{1}{2}\gamma_R^{\hat\alpha}\gamma_R^{\hat\beta}\Gamma^\mu_{\hat\alpha\hat\beta}\frac{\partial}{\partial c^\mu}
 \\ \hphantom{{\rm d}_{\mathfrak{sP}} =}{}
 + c^{[\mu\nu]} \biggl(X^{\mu}{\partial\over \partial X^{\nu}} +
 \biggl(\theta_L\Gamma_{\mu\nu}{\partial\over\partial\theta_L}\biggr) +
 \biggl(\theta_R\Gamma_{\mu\nu}{\partial\over\partial\theta_R}\biggr)\biggr)
 \end{gather*}
and
\begin{equation*}
 Q = \chi_L\lambda^\alpha_L\frac{\partial}{\partial\theta^\alpha_L} + \frac{\chi_L}{2}(\lambda_L\Gamma^\mu\theta_L)\frac{\partial}{\partial X^\mu}+\chi_R\lambda^{\hat\alpha}_R\frac{\partial}{\partial\theta_R^{\hat\alpha}} + \frac{\chi_R}{2}(\lambda_R\Gamma^\mu\theta_R)\frac{\partial}{\partial X^\mu}.
 \end{equation*}

\subsubsection{Supersymmetries}

The infinitesimal supersymmetry transformations are
\begin{align*} \xi_\epsilon =
 {}& - \frac{1}{2}(\epsilon_L\Gamma^\mu\theta_L)\frac{\partial}{\partial X^\mu}
 + \epsilon_L^\alpha\frac{\partial}{\partial\theta_L^\alpha}
 + \biggl(
 \frac{1}{2}(\epsilon_L\Gamma_\mu)_\alpha\partial_+ X^\mu +
 \frac{1}{8}(\epsilon_L\Gamma^\mu\theta_L)(\partial_+\theta_L\Gamma_\mu)_\alpha
 \biggr)\frac{\partial}{\partial p_{+\alpha}}
 \\
 &{}- \frac{1}{2}(\epsilon_R\Gamma^\mu\theta_R)\frac{\partial}{\partial X^\mu}
 + \epsilon_R^{\hat{\alpha}}\frac{\partial}{\partial\theta_R^{\hat{\alpha}}}
 + \biggl(
 \frac{1}{2}(\epsilon_R\Gamma_\mu)_{\hat{\alpha}}\partial_- X^\mu +
 \frac{1}{8}(\epsilon_R\Gamma^\mu\theta_R)(\partial_-\theta_R\Gamma_\mu)_{\hat{\alpha}}
 \biggr)\frac{\partial}{\partial p_{-\hat{\alpha}}} .
 \end{align*}
The corresponding one-form of equation~(\ref{NoetherCurrent}) is
\begin{gather*}
\alpha|_{\mathfrak{sP}}\langle\gamma,c\rangle =
 \frac{1}{4}(\gamma_L\Gamma^\mu\theta_L){\rm d}X_\mu - \frac{1}{24}(\gamma_L\Gamma^\mu\theta_L)({\rm d}\theta_L\Gamma_\mu\theta_L)
 \\ \hphantom{\alpha|_{\mathfrak{sP}}\langle\gamma,c\rangle =}{}
 - \frac{1}{4}(\gamma_R\Gamma^\mu\theta_R){\rm d}X_\mu + \frac{1}{24}(\gamma_R\Gamma^\mu\theta_R)({\rm d}\theta_R\Gamma_\mu\theta_R).
 \end{gather*}

\subsubsection{BRST transformations}

The BRST transformation is the following vector field on the space of worldsheet fields:
\begin{align*}
 Q ={}& \lambda_L\frac{\delta}{\delta\theta_L} +\frac{1}{2}(\lambda_L\Gamma^m\theta_L)\frac{\delta}{\delta X^m} + {\rm d}_+\frac{\delta}{\delta w_+}
\\
 &{} +\biggl(-\frac{1}{2}\partial_+ X^m(\lambda_L\Gamma_m) +\frac{3}{8}(\lambda_L\Gamma^m\theta_L)(\partial_+\theta_L\Gamma_m)+\frac{1}{8}(\partial_+\lambda_L\Gamma^m\theta_L)(\theta_L\Gamma_m)\biggr)\frac{\delta}{\delta p_+}
\\
 &{} + \lambda_R\frac{\delta}{\delta\theta_R} + \frac{1}{2}(\lambda_R\Gamma^m\theta_R)\frac{\delta}{\delta X^m} + {\rm d}_-\frac{\delta}{\delta w_-}
\\
 &{} + \biggl( -\frac{1}{2}\partial_- X^m(\lambda_R\Gamma_m) +\frac{3}{8}(\lambda_R\Gamma^m\theta_R)(\partial_-\theta_R\Gamma_m)+\frac{1}{8}(\partial_-\lambda_R\Gamma^m\theta_R)(\theta_R\Gamma_m) \biggr)\frac{\delta}{\delta p_-}.
 \end{align*}
The corresponding one-form of equation~(\ref{NoetherCurrent}) is
\begin{align*} \alpha|_{\mathbb{R}^{0|2}}\langle\chi\rangle =
 {}& -\frac{1}{4}\chi_L(\lambda_L\Gamma^m\theta_L)\biggl({\rm d}X_m-\frac{1}{2}\theta_L\Gamma_m{\rm d}\theta_L\biggr)
 \\
 &{} +\frac{1}{4}\chi_R(\lambda_R\Gamma^m\theta_R)\biggl({\rm d}X_m-\frac{1}{2}\theta_R\Gamma_m{\rm d}\theta_R\biggr).
 \end{align*}
It can be represented as
\begin{align} & \alpha\langle\chi\rangle=
 \iota_Q b,
 \label{AlphaChiVsB}
\end{align}
where
\begin{gather}
b = - {1\over 4}
 \left(
 ({\rm d}\theta_L\Gamma^m\theta_L) - ({\rm d}\theta_R\Gamma^m\theta_R)
 \right)
 {\rm d}X^m + {1\over 8}({\rm d}\theta_L\Gamma^m\theta_L)({\rm d}\theta_R\Gamma^m\theta_R),
 \label{DefBFlatSpace}
 \\
{\rm d}b = - {1\over 4}(({\rm d}\theta_L\Gamma^m{\rm d}\theta_L) - ({\rm d}\theta_R\Gamma^m{\rm d}\theta_R)){\rm d}X^m
 + {1\over 8}({\rm d}\theta_L\Gamma^m{\rm d}\theta_L)({\rm d}\theta_R\Gamma^m\theta_R) \nonumber
 \\ \hphantom{{\rm d}b =}{}- {1\over 8}({\rm d}\theta_L\Gamma^m\theta_L)({\rm d}\theta_R\Gamma^m{\rm d}\theta_R)\nonumber
 \\ \hphantom{{\rm d}b}{} =
 {1\over 4}(({\rm d}\theta_L\Gamma^m{\rm d}\theta_L) - ({\rm d}\theta_R\Gamma^m{\rm d}\theta_R))
 \biggl(-{\rm d}X^m + \frac{1}{2}({\rm d}\theta_L\Gamma^m\theta_L) - \frac{1}{2}({\rm d}\theta_R\Gamma^m\theta_R)\biggr).\nonumber
\end{gather}
Therefore, $\iota_Q^2 {\rm d} b = 0$ and equation~(\ref{DLieVsB}) implies that equation~(\ref{SecondConditionForCovariance}) is satisfied
and therefore the BRST current is BRST closed:
\begin{equation*}
 Qj\langle Q\rangle = 0.
 \end{equation*}
This is also true for general curved backgrounds, and is a consequence of the ghost number symmetry,
see equation~(\ref{UniversalVasIotaAlpha}).

\subsection{Ascending in form rank}

Because of Shapiro's lemma,
\begin{equation}
 H^1\bigl(\mathfrak{sP},\Omega^2(X)\bigr) = 0.
 \label{ShapiroOmega2}\end{equation}
Therefore, exists $b\in \Omega^2(X)$ such that
\begin{align} {\rm d}\alpha\langle\xi\rangle =
 {\cal L}_{\xi} b\qquad\mbox{for}\ \xi\in\mathfrak{sP}.
 \label{BField} \end{align}
This is the same $b$ as in equation~(\ref{DefBFlatSpace}), except equation~(\ref{BField}) defines it unambiguously
(while equation~(\ref{DefBFlatSpace}) only defines $b$ up to an expression having zero contraction with $Q$).
Its exterior derivative $H = {\rm d}b$ is super-Poincar\'e invariant:
\begin{equation*}
 {\cal L}_{\xi} H = 0.
 \end{equation*}
Equations~(\ref{AlphaChiVsB}) and (\ref{BField}) together imply for $\xi\in\mathfrak{sP}$ and $\chi\in\mathbb{R}^{0|2}$,
\begin{equation*}
 V\langle\xi\wedge\chi\rangle = \iota_{\chi} \alpha\langle\xi\rangle.
 \end{equation*}
This is a particular case of the ``asymmetric covariance'', equation~(\ref{AsymmetricCovarianceCondition}).
Namely, the BRST currents are super-Poincar\'e-invariant, but the super-Poincar\'e currents are not
BRST-closed, see Section \ref{sec:sPNotBRSTClosed}.

\subsection{Descending in form rank}

We will use the Faddeev--Popov language. Let $\gamma_L^{\alpha}$,
$\gamma_R^{\hat{\alpha}}$ be the ghosts for left and right supersymmetries,
$c^{\mu}$ and $c^{[\mu\nu]}$ the ghosts for translations and rotations.
We will also consider the BRST symmetries. The ghosts for BRST symmetries
will be denoted $\chi_L$ and $\chi_R$. The descent goes as follows:
\begin{gather}
 {\rm d}_{\rm Lie}\alpha={\rm d}V, \nonumber
 \\ V =
 \frac{1}{8}(\gamma_L\Gamma^\mu\theta_L)(\gamma_R\Gamma_\mu\theta_R)
 + \frac{1}{4}(\gamma_L\Gamma_{\mu}\gamma_L)X^{\mu} - \frac{1}{4}(\gamma_R\Gamma_{\mu}\gamma_R)X^{\mu} -\frac{\chi_L}{8}(\lambda_L\Gamma^\mu\theta_L)(\gamma_R\Gamma_\mu\theta_R) \nonumber
 \\ \hphantom{V =}{}
 +\frac{\chi_R}{8}(\lambda_R\Gamma^\mu\theta_R)(\gamma_L\Gamma_\mu\theta_L)
 -\frac{\chi_R}{12}(\lambda_R\Gamma^\mu\theta_R)(\gamma_R\Gamma_\mu\theta_R)+\frac{\chi_L}{12}(\lambda_L\Gamma^\mu\theta_L)(\gamma_L\Gamma_\mu\theta_L) \nonumber
 \\ \hphantom{V =}{}
 - \frac{\chi_L\chi_R}{8}(\lambda_L\Gamma^\mu\theta_L)(\lambda_R\Gamma_\mu\theta_R),
 \label{FlatSpaceV}
 \\
 {\rm d}_{\rm Lie} V = W, \nonumber
 \\ W =
 {1\over 4}(\gamma_L\Gamma_{\mu}\gamma_L)c^{\mu} - {1\over 4}(\gamma_R\Gamma_{\mu}\gamma_R)c^{\mu}.
 \label{WInFlatSpace}
 \end{gather}
The last line of $V$ contains (as the coefficient of $\chi_L\chi_R$) the universal vertex operator of flat space,
i.e., $(\lambda_L\Gamma^\mu\theta_L)(\lambda_R\Gamma_\mu\theta_R)$.

The 3-cocycle $W$ is nontrivial in $H^3(\mathfrak{sP})\subset H^3\bigl(\mathfrak{sP}\oplus {\bf R}^{0|2}_{\rm BRST}\bigr)$, it
does not involve $\chi_L$ and $\chi_R$.
In fact, it could not possibly contain $\chi_L$ and $\chi_R$, because it is constant (does not depend on~$X$, $\theta_L$, $\theta_R$). The degree of $\chi$ equals the ghost number. Expressions with nonzero ghost number can not be constant.

\subsection[Relation between H=db and W]{Relation between $\boldsymbol{H={\rm d}b}$ and $\boldsymbol{W}$}

Now we have the relation
\begin{equation*}
 H + W = ({\rm d} + {\rm d}_{\mathfrak{sP}}) (b - \alpha_{\mathfrak{sP}} + V_{\mathfrak{sP}})
 \end{equation*}
and its equivariant analogue
\begin{equation*}
 H + \lambda_{\mathfrak{sP}}+\Psi_{\mathfrak{sP}} + W =
 ({\rm d} + {\rm d}_{\mathfrak{sP}} + \mathbb{I}_{\mathfrak{sP}} + \mathbb{J}_{\mathfrak{sP}})
 (b - \alpha_{\mathfrak{sP}} + V_{\mathfrak{sP}}).
 \end{equation*}
Such relation exists because $b$ exists because of equation~(\ref{ShapiroOmega2}).
It is limited to $\mathfrak{sP}\subset \mathfrak{sP}\oplus \mathbb{R}^{0|2}$, because
there is ``no $b$ for $Q$''. There is no such $b$ which would satisfy ${\cal L}_Q b = {\rm d}\alpha\langle Q\rangle$.
The $H^1_Q\bigl(\Omega^2(X)\bigr)$ is nonzero.

\subsection[Non-covariance of sP currents]{Non-covariance of $\boldsymbol{\mathfrak{sP}}$ currents}

\subsubsection[sP currents are not sP-covariant]{$\boldsymbol{\mathfrak{sP}}$ currents are not $\boldsymbol{\mathfrak{sP}}$-covariant}

This is because $\mathbb{I}W$ is not zero in $H^2 (\mathfrak{g},\operatorname{Hom}(\adg, \mathbb{R}) )$,
see equation~(\ref{IWObstacle}). This is different from the case of WZW model.
In the case of $WZW$ model, $\mathbb{I}W$ is zero in $H^2 (\mathfrak{g},\operatorname{Hom}(\adg, \mathbb{R}) )$,
due to the existence of the invariant $\mathfrak{C}$, see equation~(\ref{PsiVsC}).

\subsubsection[sP currents are not BRST closed]{$\boldsymbol{\mathfrak{sP}}$ currents are not BRST closed}\label{sec:sPNotBRSTClosed}

Since $W$ is only nontrivial in $H^3(\mathfrak{sP})\subset H^3\bigl(\mathfrak{sP}\oplus {\bf R}^{0|2}_{\rm BRST}\bigr)$,
the corresponding component of $\mathbb{I}W$ is zero. In this case, the next obstacle, equation~(\ref{PsiObstacle}),
is nontrivial:
\begin{align*} &\Psi\in
 H^1_Q(\operatorname{Hom}_{\mathfrak{sP}}({\rm ad}_{\mathfrak{sP}}, C^{\infty}(X))).
 \end{align*}
Let us use
\begin{align*}
 & \Psi\langle \chi_L Q_L + \chi_R Q_R\rangle\langle \epsilon_L Q_L + \epsilon_R Q_R\rangle = 0,
 \\
 & (\chi_L Q_L + \chi_R Q_R)\Psi\langle \chi_L Q_L + \chi_R Q_R \rangle\langle \_|_{\mathfrak{sP}}\rangle = 0.
 \end{align*}
The first follows from the BRST closedness of the BRST current, and the second from
$\mathbb{I}W\langle Q\bullet Q\rangle = 0$.
Also, we use
\begin{equation*}
 (\mathbb{I}W)\langle C_{\mathfrak{sP}}\bullet C_{\mathfrak{sP}}\rangle\langle \_|_{\mathfrak{sP}}\rangle = 0.
 \end{equation*}
\big(We hope that the notations are not too confusing;
 $C_{\mathfrak{sP}}$ is the ghost for $\mathfrak{sP}$, and $\_|_{\mathfrak{sP}}$
 means that the argument (the $\eta$) is restricted to $\mathfrak{sP}\subset\mathfrak{sP}\oplus\mathbb{R}^{0|2}$.\big)
It implies
\begin{equation*}
 Q\Psi\langle C_{\mathfrak{sP}}\rangle \langle \_|_{\mathfrak{sP}}\rangle =
{\rm d}_{\mathfrak{sP}}(\Psi\langle Q\rangle \langle \_|_{\mathfrak{sP}}\rangle ).
 \end{equation*}
This implies that
\begin{gather*}
 \Psi\langle \chi_L Q_L + \chi_R Q_R \rangle\langle \_|_{\mathfrak{sP}}\rangle
 \mod (\chi_L Q_L + \chi_R Q_R)(\ldots)
 \\ \qquad\quad\in
 H^0\bigl(
 \mathfrak{sP},
 \operatorname{Hom}\bigl({\rm ad}_{\mathfrak{sP}}, H^1\bigl(\mathbb{R}^{0|2}, C^{\infty}(X)\bigr)\bigr)
 \bigr).
 \end{gather*}
In other words, it is an intertwiner between the adjoint of $\mathfrak{sP}$ and the ghost number one BRST cohomology.
This intertwiner associates to each element of $\mathfrak{sP}$ the corresponding ghost number one vertex operator.
It is enough to compute $\Psi\langle Q\rangle$ on an infinitesimal translation,
i.e., $\eta = {\partial\over\partial X^m}$. In this case, the $\mathbb{I}V$ in
$\Psi = \mathbb{I}V - \mathbb{J}\alpha$ does not contribute, and we are left with
$-\mathbb{J}\alpha$ which gives
\begin{equation*}
 \iota_{\partial/\partial X^m}\alpha\langle Q\rangle =
 - {1\over 4}\chi_L(\lambda_L\Gamma_m\theta_L) + {1\over 4}\chi_R(\lambda_R\Gamma_m\theta_R).
 \end{equation*}
Instead of being BRST closed, they satisfy
\begin{gather*}
Qj\langle\eta\rangle=
 {\rm d}\Lambda\langle\eta\rangle,
 \end{gather*}
where
\[ \Lambda\langle\eta\rangle=
 \Psi\langle Q\rangle\langle \eta\rangle
 = V\langle \eta\wedge Q \rangle + \iota_{\eta}\alpha\langle Q\rangle
 = \iota_{\eta}\alpha\langle Q\rangle + \iota_Q\alpha\langle \eta\rangle.
\]

\subsection[Cohomology of Q in one-forms]{Cohomology of $\boldsymbol{Q}$ in one-forms}\label{sec:HQinOneForms}

The last term in equation~(\ref{FlatSpaceV}) is BRST exact:
\begin{equation*}
 \frac{\chi_L\chi_R}{8}(\lambda_L\Gamma^\mu\theta_L)(\lambda_R\Gamma_\mu\theta_R)
 =
 (\chi_L Q_L + \chi_R Q_R)
 \bigl(
 {1\over 8} X^m (\chi_R(\lambda_R\Gamma^m\theta_R) - \chi_L(\lambda_L\Gamma^m\theta_R))
 \biggr).
 \end{equation*}
This implies that we can add a total derivative to $\alpha\langle Q\rangle$ so that the modified
(improved) $\alpha\langle Q\rangle$ is BRST-closed. It is given by the following expression,
which contains ``bare'' $X^m$:
\begin{gather*}
 \alpha_{\rm impr}\langle Q\rangle =
 {1\over 8}
 ( \chi_L (\lambda_L \Gamma^m \theta_L) {\rm d}X^m + \chi_L ({\rm d}(\lambda_L \Gamma^m \theta_L)) X^m
 \\ \hphantom{ \alpha_{\rm impr}\langle Q\rangle =}{}
 - \chi_L (\lambda_L \Gamma^m \theta_L)(\theta_L \Gamma^m {\rm d}\theta_L) - (L \leftrightarrow R)).
\end{gather*}
We still can not make the Lagrangian BRST-invariant, since $\alpha_{\rm impr}$ represents a nontrivial cohomology class:
\begin{equation*}
 [\alpha_{\rm impr}] \in H^1\bigl(Q,\Omega^1(X)\bigr).
 \end{equation*}
Moreover, the currents do not transform covariantly.
Indeed, the right-hand side of equation~(\ref{DeviationFromCovariance}) becomes ${\rm d} \iota_Q \alpha_{\rm impr}\langle Q\rangle$ with
\begin{equation*}
 \iota_Q \alpha_{\rm impr}\langle Q\rangle = - {1\over 8} \chi_L\chi_R (\lambda_L\Gamma^m\theta_L)(\lambda_R\Gamma^m\theta_R) \neq 0.
 \end{equation*}
On-shell $\alpha_{\rm impr}$ is BRST-equivalent to
${1\over 8}\chi_L(\lambda_L\Gamma^m\theta_L)\bigl(\partial X^m - \frac{1}{2}(\theta_L\Gamma^m\partial\theta_L)\bigr) - (L\leftrightarrow R)$.
In a general curved background, the universal vertex operator is a nontrivial
element of $H^2(Q,C^{\infty}(X))$. But flat space is a special case.
What would be a ghost-number-two universal vertex is now BRST exact. Instead, what prevents the Lagrangian from
being BRST invariant is a nontrivial element of $H^1\bigl(Q,\Omega^1(X)\bigr)$.

\subsection{Covariance with respect to a smaller subalgebra}

We can achieve covariance with respect to a smaller subalgebra $\mathfrak{h}\subset \mathfrak{sP}$,
if the restriction of~$W$ to~$\mathfrak{h}$ vanishes. For example, let $\mathfrak{h}$ be the
subalgebra preserving the half-BPS state, then $W$ vanishes. In fact,
\begin{equation*}
 V\langle\xi\wedge\eta\rangle = ({\rm d}_{\rm Lie} U)\langle\xi\wedge\eta\rangle
\qquad \text{and}\qquad
{\rm d}_{\rm Lie}(\alpha - {\rm d}U) = 0.
 \end{equation*}
This implies
\begin{equation*}
 \alpha\langle\xi\rangle = {\rm d}U\langle\xi\rangle + {\cal L}_{\xi} f.
 \end{equation*}
Then, the modified Lagrangian $L - f$ is $\mathfrak{h}$-invariant.

We will now show this explicitly. We will denote $X^{\pm} = X^0 \pm X^9$. We will choose the momentum
of the half-BPS state in direction $X^+$, then the following bosonic generators annihilate the state:
$ {\partial\over\partial X^-}$, ${\partial\over\partial X^i}$,
where $i$ runs from $1$ to $8$.
The ten-dimensional gamma-matrices are
\begin{alignat*}{4}
& \Gamma^i_{\bullet\bullet} = \begin{pmatrix}0 & \sigma^i_{a\dot{b}} \cr \sigma^i_{\dot{a}b} & 0\end{pmatrix},\qquad&&
 \Gamma^+_{\bullet\bullet} = \begin{pmatrix}0 & 0 \cr 0 & \delta_{\dot{a}\dot{b}}\end{pmatrix},\qquad&&
 \Gamma^-_{\bullet\bullet} = \begin{pmatrix} \delta_{ab} & 0 \cr 0 & 0 \end{pmatrix},&
\\
& \Gamma^{i\bullet\bullet} = \begin{pmatrix}0 & -\sigma^{ia\dot{b}} \cr -\sigma^{i\dot{a}b} & 0\end{pmatrix},\qquad&&
 \Gamma^{+\bullet\bullet} = \begin{pmatrix} \delta^{ab} & 0 \cr 0 & 0\end{pmatrix},\qquad&&
 \Gamma^{-\bullet\bullet} = \begin{pmatrix} 0 & 0 \cr 0 & \delta^{\dot{a}\dot{b}} \end{pmatrix}.&
\end{alignat*}
We have to restrict
\begin{equation*}
 c^+ = 0,\qquad \gamma_L^{\dot{a}} = 0,\qquad \gamma_R^{\dot{a}}=0.
 \end{equation*}
The first line of equation~(\ref{FlatSpaceV}) becomes
\begin{align*}
& V=
 - {1\over 8} \bigl(\gamma_L^a\sigma^i_{a\dot{a}}\theta_L^{\dot{a}}\bigr)\bigl(\gamma_R^a\sigma^i_{a\dot{a}}\theta_R^{\dot{a}}\bigr)
 + {1\over 4} \bigl(\gamma_L^a\sigma^i_{a\dot{a}}\theta_L^{\dot{a}}\bigr) c^i
 - {1\over 4} \bigl(\gamma_R^a\sigma^i_{a\dot{a}}\theta_R^{\dot{a}}\bigr) c^i.
 \end{align*}
It is ${\rm d}_{\rm Lie}$-exact:
\[
V = {\rm d}_{\rm Lie}U,
\]
where
\begin{gather*}
 U=\biggl( - {1\over 16} \bigl(\theta_L^a\sigma^i_{a\dot{a}}\theta_L^{\dot{a}}\bigr)\bigl(\gamma_R^a\sigma^i_{a\dot{a}}\theta_R^{\dot{a}}\bigr)
 + {1\over 16} \bigl(\gamma_L^a\sigma^i_{a\dot{a}}\theta_L^{\dot{a}}\bigr)\bigl(\theta_R^a\sigma^i_{a\dot{a}}\theta_R^{\dot{a}}\bigr)
 \\ \hphantom{U=\biggl(}{}
 + {1\over 4} \bigl(\theta_L^a\sigma_{a\dot{a}}^i\theta_L^{\dot{a}}\bigr)c^i
 - {1\over 4} \bigl(\theta_R^a\sigma_{a\dot{a}}^i\theta_R^{\dot{a}}\bigr)c^i
 \biggr)
 \end{gather*}
and
\begin{align*}
\alpha=
 &{} {\rm d}U + {\rm d}_{\rm Lie}f,
 \\ f=
 &{} -\frac{1}{8}\bigl(\theta_L^a{\rm d}\theta_L^a-\theta_R{\rm d}\theta_R\bigr)X_+ + \frac{1}{4}\bigl(\theta_L^a\sigma^i_{a\dot a}\theta_L^{\dot a}\bigr){\rm d}X_i- \frac{1}{4}\bigl(\theta_R^a\sigma^i_{a\dot a}\theta_R^{\dot a}\bigr){\rm d}X_i \\
 &{} - \frac{1}{8}\bigl(\theta_L^a\sigma^i_{a\dot a}\theta_L^{\dot a}\bigr)\bigl({\rm d}\theta_L^a\sigma^i_{a\dot a}\theta^{\dot a}_L\bigr) + \frac{1}{8}\bigl(\theta_R^a\sigma^i_{a\dot a}\theta_R^{\dot a}\bigr)\bigl({\rm d}\theta_R^a\sigma^i_{a\dot a}\theta^{\dot a}_R\bigr) \\
 &{} - \frac{1}{8}\bigl(\theta_L^a\sigma^i_{a\dot a}\theta_L^{\dot a}\bigr)\bigl({\rm d}\theta_L^{\dot a}\sigma^i_{a\dot a}\theta^{a}_L\bigr) + \frac{1}{8}\bigl(\theta_R^a\sigma^i_{a\dot a}\theta_R^{\dot a}\bigr)\bigl({\rm d}\theta_R^{\dot a}\sigma^i_{a\dot a}\theta^{a}_R\bigr) \\
 &{} - \frac{1}{16}\bigl(\theta_L^a\sigma^i_{a\dot a}\theta_L^{\dot a}\bigr)\bigl(\theta_R^a\sigma^i_{a\dot a}{\rm d}\theta_R^{\dot a}\bigr) - \frac{1}{16}\bigl(\theta_R^a\sigma^i_{a\dot a}\theta_R^{\dot a}\bigr)\bigl({\rm d}\theta_L^a\sigma^i_{a\dot a}\theta_L^{\dot a}\bigr)
 \\
 &{} + \frac{1}{16}\bigl(\theta_R^a\sigma^i_{a\dot a}\theta_R^{\dot a}\bigr)\bigl(\theta_L^a\sigma^i_{a\dot a}{\rm d}\theta_L^{\dot a}\bigr) + \frac{1}{16}\bigl(\theta_L^a\sigma^i_{a\dot a}\theta_L^{\dot a}\bigr)\bigl({\rm d}\theta_R^a\sigma^i_{a\dot a}\theta_R^{\dot a}\bigr).
 \end{align*}

\section[Pure spinor sigma-model in AdS\_5 times S\^{}5]{Pure spinor sigma-model in $\boldsymbol{{\rm AdS}_5\times S^5}$}\label{PureSpinorAdS}

\subsection[Lagrangian is invariant under psu(2,2|4)]{Lagrangian is invariant under $\boldsymbol{\mathfrak{psu}(2,2|4)}$}

This is because $H^3(\mathfrak{psu}(2,2|4))=0$. It is useful to compare this to flat
space. The flat space limit of $\mathfrak{psu}(2,2|4)$
is $\mathfrak{sP}_{5+5}\subset\mathfrak{sP}$ -- the subalgebra of the
super-Poincar\'e algebra, excluding those Lorentz rotations which do not preserve
the $5+5$ split of the tangent space.
We can consider the differential $\mathfrak{psu}(2,2|4)$ as deformed differential
of $\mathfrak{sP}_{5+5}$:
\begin{equation*}
 {\rm d}_{\mathfrak{psu}(2,2|4)} = {\rm d}_{\mathfrak{sP}_{5+5}} + {\rm d}',
 \end{equation*}
where $d'$ is a correction. Then $W\in H^3(\mathfrak{sP}_{5+5})$ gets
cancelled by an element of $H^2(\mathfrak{sP}_{5+5})$:
\begin{equation*}
 W = {\rm d}' {\rm STr}(\gamma_L\gamma_R).
 \end{equation*}

\subsection{Universal vertex operator}

\begin{align*}
&\alpha\langle \xi\rangle =
 0\qquad\text{for}\ \xi\in\mathfrak{psu}(2,2|4),
\\
& \alpha\langle Q_L\rangle =
 {\rm STr}\bigl(\bigl({\rm d}g g^{-1}\bigr)_{\bar{1}}\lambda_{\bar{3}}\bigr),
 \qquad
 \alpha\langle Q_R\rangle =
 - {\rm STr}\bigl(\bigl({\rm d} g g^{-1}\bigr)_{\bar{3}}\lambda_{\bar{1}}\bigr).
\end{align*}
Therefore,
\begin{align*}& V\langle Q_L\wedge Q_L\rangle=
 V\langle Q_R\wedge Q_R\rangle = 0
 ,\qquad V\langle Q_L\wedge Q_R\rangle=
 {\rm STr}(\lambda_{\bar{3}}\lambda_{\bar{1}}).
 \end{align*}

The currents of the global symmetries are only $Q$-invariant up to a total derivative:
\[
 Qj_a= {\rm d}\Lambda_a,
\qquad
\Lambda_a = {\rm d}\bigl({\rm STr}\bigl(t_ag^{-1}(\lambda_3-\lambda_1)g\bigr)\bigr).
\]
But both BRST current and $\alpha\langle Q\rangle$ are strictly invariant under the global symmetries.
Therefore, equation~(\ref{CocycleExt}) implies
\begin{equation*}
 \Lambda_a = \iota_{t_a}\alpha\langle Q\rangle.
\end{equation*}
This can be directly verified:
\begin{equation*}
 {\rm STr}\bigl(t_ag^{-1}(\lambda_3-\lambda_1)g\bigr) =
 \iota_{t_a}{\rm STr}\bigl(\lambda_3 \bigl(dgg^{-1}\bigr)_1 - \lambda_1 \bigl(dgg^{-1}\bigr)_3\bigr).
\end{equation*}
This can be also interpreted in the following way.
Since ${\cal L}_{\xi}\alpha\langle Q\rangle=0$, we have
\begin{equation*}
{\rm d}\iota_{\xi}\alpha\langle Q\rangle = - \iota_{\xi} QL.
 \end{equation*}
Here $QL$ is the pullback of an exact two-form from the target space,
and $\iota_{\xi} QL$ is (by a slight abuse of notations) the pullback of its contraction with $\xi$.
Therefore,
\begin{equation*}
 \iota_{\xi}QL = Qj\langle\xi\rangle.
\end{equation*}

\subsection{``Minimalistic'' B-field}

According our general scheme,
${\rm d} (\chi_L \alpha\langle Q_L\rangle + \chi_R \alpha\langle Q_R\rangle )$
is the variation of the Lagrangian under $\chi_L Q_L + \chi_R Q_R$.
We would like to stress that the Lagrangian can not be seen as a geometrical
object on the target space.
The Lagrangian is not a pullback of a differential form on the
worldsheet, as it contains terms like ${\rm d}X\wedge *{\rm d}X$ and $p\bar{\partial}X$.
But, one can ask: is it \emph{possible} to find
such a 2-form $\cal B$ on the target space that its Lie derivative
along $\chi_L Q_L + \chi_R Q_R$ be equal to that total derivative?
The answer is negative, in fact
${\rm d}(\chi_L \alpha\langle Q_L\rangle + \chi_L \alpha\langle Q_L\rangle)$
is a nontrivial element of
$H^2\bigl(\mathbb{R}^{0|2}_{\rm BRST},\Omega^2(X)\bigr)$.
But, it is possible to find such $\cal B$ if we do not require it to be
smooth. In fact,
\begin{gather*}
{\rm d}(\chi_L \alpha\langle Q_L\rangle + \chi_R \alpha\langle Q_R\rangle)=
 {\cal L}_{\chi_L Q_L + \chi_R Q_R} {\cal B},
\qquad
 {\cal B} = {\rm STr}(J_3(1-2\mathbb{P}_{31})J_1)
 \end{gather*}

Here $\mathbb{P}_{31}$ is some projector, which is a rational function of $\lambda_L$ and $\lambda_R$.
See \cite{Mikhailov:2017mdo} for details.
We will here list some properties of $\cal B$ and ${\cal H} = {\rm d}{\cal B}$:
\begin{align*}
 & {\cal H} = {\rm d}{\cal B},
\qquad \iota_Q{\cal H} = 0,
\qquad {\rm e}^{\iota_Q}{\cal B} = {\cal B} - \alpha + V,
\qquad {\cal H} = ({\rm d} + {\cal L}_Q)({\cal B} - \alpha + V),
 \end{align*}
where $Q = \chi_LQ_L + \chi_RQ_R$. This is a particular case of equation~(\ref{Kalkman}), when the Lie superalgebra is~$\mathbb{R}^{0|2}$.
The $C$-ghosts are $\chi_L$ and $\chi_R$, and ${\rm d}_{\rm Lie}^{(0)}=0$ because the Lie superalgebra
$\mathbb{R}^{0|2}$ is abelian.

Moreover,
\begin{equation*}
 (\mathbb{I} + \mathbb{J})({\cal B} - \alpha + V) = 0.
 \end{equation*}
Therefore, $\cal H$ is a coboundary in the BRST model of equivariant cohomology:
\begin{equation*}
 {\cal H} = ({\rm d} + {\rm d}_{{\rm Lie}\mathbb{R}^{0|2}} + \mathbb{I} + \mathbb{J})({\cal B} - \alpha + V).
 \end{equation*}
This can be extended to the full $\mathfrak{g} = \mathbb{R}^{0|2}\oplus \mathfrak{psu}(2,2|4)$, giving an analogue of equation~(\ref{DtotWZW}):
\begin{equation*}
 {\cal H} + \mathbb{J}_{\mathfrak{psu}}{\cal B} + \Psi = ({\rm d} + {\rm d}_{\rm Lie} + \mathbb{I} + \mathbb{J})({\cal B} - \alpha_{\mathbb{R}^{0|2}} + V),
 \end{equation*}
where the nonzero component of $\Psi$ is
\begin{align*}
& \Psi\langle \chi_L,\chi_R\rangle\langle \eta\rangle =
 {\rm STr}\bigl(\eta g^{-1}(\chi_L\lambda_L - \chi_R\lambda_R)g\bigr),
 \qquad \eta\in\mathfrak{psu}(2,2|4). 
 \end{align*}
Here $\mathbb{J}_{\mathfrak{psu}}$ is the contraction with $\eta\in \mathfrak{psu}(2,2|4)$.

The 2-form $\cal B$ is the key ingredient in the BV formulation of the AdS sigma-model developed in
\cite{Berkovits:2008ga,Berkovits:2019ulm,Mikhailov:2017mdo}. The one-form $\mathbb{J}_{\mathfrak{psu}}{\cal B}$
satisfies the analogue of equation~(\ref{IotaHisExact}) for $\lambda$ of the WZW model:
\begin{equation*}
{\rm d}\mathbb{J}_{\mathfrak{psu}}{\cal B} + \mathbb{J}_{\rm psu}{\cal H} = 0.
 \end{equation*}
In BV formalism of \cite{Mikhailov:2017mdo}, $\iota_v{\cal H} = \{S_{\rm BV},\underline{v}\}$
for any vector field $v$ on the target space, where ${\underline{v} = v^{\mu}x^{\star}_{\mu}}$
is the corresponding BV Hamiltonian.
Therefore, $\mathbb{J}_{\mathfrak{psu}}{\cal B}$ can be interpreted as a ``topological current''.
We hope to further investigate the properties of these currents in a future publication.

\section[Flat space limit of AdS\_5 times S\^{}5]{Flat space limit of $\boldsymbol{{\rm AdS}_5\times S^5}$}\label{FlatSpaceLimitOfAdS}

Consider pure spinor superstring sigma-model with the target space ${\rm AdS}_5\times S^5$ of the radius $R$.
If the motion of the string is limited to a neighborhood of a point, then in the limit $R\rightarrow \infty$
the space-time is effectively flat. The details of how the limit is taken were worked out, e.g.,
in~\cite{Mikhailov:2012id}. Here we will address the following question: why is it that the Lagrangian
in flat space is only invariant up to a total derivative, while in AdS it is exactly invariant?
This can be understood in the following way. First of all, $\mathfrak{psu}(2,2|4)$ is actually not a deformation
of $\mathfrak{sP}$, but rather the deformation of a subalgebra $\mathfrak{sP}_{5+5}\subset\mathfrak{sP}$
(the one which preserves the constant RR 5-form $F_5$). Those Lorentz transformations which do not preserve
$F_5$ are broken in~AdS, they do not correspond to any elements of $\mathfrak{psu}(2,2|4)$.
All other elements survive deformed, so that $\mathfrak{sP}_{5+5}\subset\mathfrak{sP}$ gets deformed
to $\mathfrak{psu}(2,2|4)$. Unlike $\mathfrak{sP}$, $\mathfrak{sP}_{5+5}$ has a nontrivial second
cohomology group $H^2(\mathfrak{sP}_{5+5}, \mathbb{R})$, see equation~(\ref{H2ofSp55}).
At the same time, ${\rm d}_{\rm Lie}$ gets deformed.
It turns out that the obstacle $W$ defined in
equation~(\ref{WInFlatSpace}) is cancelled by the deformation of ${\rm d}_{\rm Lie}$ acting on that element of
$H^2(\mathfrak{sP}_{5+5}, \mathbb{R})$.

We will now describe in detail how this happens.

\subsection{Zooming on a point in AdS}

The scale of the fundamental fields $\theta$ and $x$ are
\begin{equation*}
 [x]=R^{-1},\qquad [\theta]=[\lambda]=R^{-1/2},\qquad [p]=[w]=R^{-3/2}.
 \end{equation*}
The action up to order $R^{-2}$ is
\begin{align*} S =
 &{}\int {\rm d}^2\tau\, R^{-1}L_1+R^{-2}(L_2+L_3+L_4)
 \\ =
 &{}\int {\rm d}^2\tau \biggl[R^{-1}\partial_+\theta_R\partial_-\theta_L + R^{-2}\biggl(\frac{1}{2}\partial_+x\partial_-x + L_3+L_4\biggr)\biggr]
 \end{align*}
with
\begin{align*}& L_3 =
 -\frac{1}{2}([\theta_R,\partial_+\theta_R],\partial_-x) -\frac{1}{2}(\partial_+x,[\theta_L,\partial_-\theta_L]),
\\
 & L_4 =
 - \frac{1}{24}([\theta_L,\partial_+\theta_L],[\theta_L,\partial_-\theta_L])-\!\frac{1}{24}([\theta_R,\partial_+\theta_R],[\theta_R,\partial_-\theta_R])
 -\!\frac{1}{12}([\theta_R,\partial_+\theta_R],[\theta_L,\partial_-\theta_L])
 \\
 &\hphantom{L_4 =}{} - \frac{1}{6}([\theta_R,\partial_+\theta_L],[\theta_R,\partial_-\theta_L]) -\!\frac{1}{6}([\theta_L,\partial_+\theta_R],[\theta_L,\partial_-\theta_R]) -\!\frac{1}{3}([\theta_L,\partial_+\theta_R],[\theta_R,\partial_-\theta_L]).
 \end{align*}
We introduce the first order formalism in the first term by introducing $\tilde{p}$ fields, such that $p_1 = p_{1+}\in\mathfrak{g}_1$ and $p_3 = p_{3-}\in\mathfrak{g}_3$. The relevant term in the action is substituted by
\begin{equation*}
 R^{-1}\partial_+\theta_R\partial_-\theta_L \longmapsto R^{-2}(\widetilde{p}_{1+}\partial_-\theta_L) + R^{-2}(\widetilde{p}_{3-}\partial_+\theta_R) - R^{-3}\widetilde{p}_{1+}\widetilde{p}_{3-}.
 \end{equation*}
Integrating out in $\tilde{p}$, we ga back to the original Lagrangian, in which case:
\begin{align*}
 & \widetilde{p}_{1+}=R\partial_+\theta_R, \hspace{1cm} \widetilde{p}_{3-}=R\partial_-\theta_L.
 \end{align*}
Now, we redefine the $\tilde{p}$ fields:
\begin{gather*} p_{1+} =
 \widetilde{p}_{1+} + \frac{1}{2}[\theta_L,\partial_+x] +\frac{1}{24}[\theta_L,[\theta_L,\partial_+\theta_L]]+\frac{1}{24}[\theta_L,[\theta_R,\partial_+\theta_R]]
 \\ \hphantom{p_{1+} =}{}
 +\frac{1}{6}[\theta_R,[\theta_R,\partial_+\theta_L]]+\frac{1}{6}[\theta_R,[\theta_L,\partial_+\theta_R]],
 \\ p_{3-} =
 \widetilde{p}_{3-} + \frac{1}{2}[\theta_R,\partial_-x] + \frac{1}{24}[\theta_R,[\theta_R,\partial_-\theta_R]]+\frac{1}{24}[\theta_R,[\theta_L,\partial_-\theta_L]]
 \\ \hphantom{p_{3-} =}{}
 +\frac{1}{6}[\theta_L,[\theta_L,\partial_-\theta_R]]+\frac{1}{6}[\theta_L,[\theta_R,\partial_-\theta_L]].
\end{gather*}
This makes the leading order on the action to be
\begin{equation*}
 S = \int {\rm d}^2\tau\, R^{-2}\biggl[ (p_{1+}\partial_-\theta_L)+(p_{3-}\partial_+\theta_R)+\frac{1}{2}\partial_+x\partial_-x \biggr].
 \end{equation*}

\subsection[How W arises in the flat space limit]{How $\boldsymbol{W}$ arises in the flat space limit}

The obstruction $W$, which is an element of the $H^3\bigl(\mathfrak{sP}\oplus {\bf R}^{0|2}_{\rm BRST},\mathbb{C}\bigr)$ cohomology begins to appear when we observe that the flat limit is obtained from AdS limit only after integrating by parts, and then introducing the momentum $p$:
\begin{equation*}
 L_{\rm AdS}^{(p)} = L_{\rm flat}^{(p)} + {\rm d}\beta + \mathcal{O}\bigl(R^{-3}\bigr)
 \end{equation*}
such that ${\rm d}_{\rm Lie}\beta = -\alpha+{\rm d}A$, so that ${\rm d}_{\rm Lie}L_{\rm AdS}=0$.
The explicit formula is
\begin{equation*}
 \beta=-\frac{1}{4}\theta_L{\rm d}\theta_R+\frac{1}{24}[\theta_R,\theta_L][{\rm d}\theta_R,\theta_L].
 \end{equation*}
This $\beta$ satisfies
\begin{equation*}
 \alpha = -{\rm d}_{\mathfrak{psu}}\beta + {\rm d}A
 \end{equation*}
with
\begin{equation*}
 A = \frac{1}{4}[\theta_R,\gamma_R]x+\frac{1}{24}[\theta_R,\theta_L][\gamma_R,\theta_L]-\frac{1}{4}\gamma_L\theta_R.
 \end{equation*}
Therefore, ${\rm d}V={\rm d}_\mathfrak{g}\alpha={\rm d}({\rm d}_\mathfrak{g}A)$ and then
\begin{equation*}
 V = {\rm d}_\mathfrak{psu}A + C,
 \end{equation*}
where $C$ is a constant in the target space:
\begin{equation}
 C = {\rm STr}(\gamma_L\gamma_R) = \gamma^\alpha_L\gamma^{\hat\alpha}_R{\rm STr}\bigl(T^L_\alpha T^R_{\hat\alpha}\bigr).
 \label{H2ofSp55}\end{equation}
Finally,
\begin{equation*}
{\rm d}_\mathfrak{psu}C = ([\gamma_L,\gamma_L]-[\gamma_R,\gamma_R])c = W.
 \end{equation*}
See \cite{Suszek:2018bvx} for a previous discussion of this, and \cite{Mikhailov:2012id} for a somewhat similar story
with the deformation of the BRST operator.

\subsection*{Acknowledgments}

We want to think the referees for helping us to improve the paper, in particular
for pointing out errors in the original version.
This work was supported in part by ICTP-SAIFR FAPESP grant 2019/21281-4.
The work of VB was supported by the Grant Agency of the Czech Republic under the
grant EXPRO 20-25775X.
The work of AM was supported in part by CNPq grant ``Produtividade em Pesquisa'' 307191/2022-2.
The work of EV was supported by FAPESP grant 2022/00940-2.

\pdfbookmark[1]{References}{ref}
\LastPageEnding


\begin{thebibliography}{99}
\footnotesize\itemsep=0pt

\bibitem{Alekseev:2004np}
Alekseev A., Strobl T., Current algebras and differential geometry,
 \href{https://doi.org/10.1088/1126-6708/2005/03/035}{\textit{J.~High Energy Phys.}} \textbf{2005} (2005), no.~3, 035, 14~pages,
 \href{https://arxiv.org/abs/hep-th/0410183}{arXiv:hep-th/0410183}.

\bibitem{Alexandrov:1995kv}
Alexandrov M., Schwarz A., Zaboronsky O., Kontsevich M., The geometry of the
 master equation and topological quantum field theory,
 \href{https://doi.org/10.1142/S0217751X97001031}{\textit{Internat.~J.~Modern Phys.~A}} \textbf{12} (1997), 1405--1429,
 \href{https://arxiv.org/abs/hep-th/9502010}{arXiv:hep-th/9502010}.

\bibitem{Berkovits:2000fe}
Berkovits N., Super-{P}oincar\'e covariant quantization of the superstring,
 \href{https://doi.org/10.1088/1126-6708/2000/04/018}{\textit{J.~High Energy Phys.}} \textbf{2000} (2000), no.~4, 018, 17~pages,
 \href{https://arxiv.org/abs/hep-th/0001035}{arXiv:hep-th/0001035}.

\bibitem{Berkovits:2008ga}
Berkovits N., Simplifying and extending the {${\rm AdS}_5\times S^5$} pure
 spinor formalism, \href{https://doi.org/10.1088/1126-6708/2009/09/051}{\textit{J.~High Energy Phys.}} \textbf{2009} (2009), no.~9,
 051, 34~pages, \href{https://arxiv.org/abs/0812.5074}{arXiv:0812.5074}.

\bibitem{Berkovits:2019ulm}
Berkovits N., Sketching a proof of the {M}aldacena conjecture at small radius,
 \href{https://doi.org/10.1007/jhep06(2019)111}{\textit{J.~High Energy Phys.}} \textbf{2019} (2019), no.~6, 111, 14~pages,
 \href{https://arxiv.org/abs/1903.08264}{arXiv:1903.08264}.

\bibitem{Berkovits:2001ue}
Berkovits N., Howe P., Ten-dimensional supergravity constraints from the pure
 spinor formalism for the superstring, \href{https://doi.org/10.1016/S0550-3213(02)00352-8}{\textit{Nuclear Phys.~B}} \textbf{635}
 (2002), 75--105, \href{https://arxiv.org/abs/hep-th/0112160}{arXiv:hep-th/0112160}.

\bibitem{Blohmann}
Blohmann C., The homotopy momentum map of general relativity, \href{https://doi.org/10.1093/imrn/rnac087}{\textit{Int.
 Math. Res. Not.}} \textbf{2023} (2023), 8212--8250, \href{https://arxiv.org/abs/2103.07670}{arXiv:2103.07670}.

\bibitem{Cordes:1994fc}
Cordes S., Moore G., Ramgoolam S., Lectures on~{$2$}{D} {Y}ang--{M}ills theory,
 equivariant cohomology and topological field theories, \href{https://doi.org/10.1016/0920-5632(95)00434-B}{\textit{Nuclear
 Phys.~B Proc. Suppl.}} \textbf{41} (1995), 184--244, \href{https://arxiv.org/abs/hep-th/9411210}{arXiv:hep-th/9411210}.

\bibitem{Deligne:1999qp}
Deligne P., Etingof P., Freed D.S., Jeffrey L.C., Kazhdan D., Morgan J.W.,
 Morrison D.R., Witten E. (Editors), Quantum fields and strings: {A} course
 for mathematicians, Vol.~1,~2, American Mathematical Society, 1999.

\bibitem{FeiginFuchs}
Feigin B.L., Fuchs D.B., Cohomology of {L}ie groups and {L}ie algebras, in
 Current Problems in Mathematics. {F}undamental Directions, {V}ol.~21, Itogi
 Nauki i Tekhniki, Akad. Nauk SSSR, Vsesoyuz. Inst. Nauchn. i Tekhn. Inform.,
 Moscow, 1988, 121--209.

\bibitem{GelfandManin}
Gelfand S.I., Manin Yu.I., Methods of homological algebra, 2nd ed., \textit{Springer Monogr. Math.}, \href{https://doi.org/10.1007/978-3-662-12492-5}{Springer}, Berlin, 2003.

\bibitem{Knapp}
Knapp A.W., Lie groups, {L}ie algebras, and cohomology, \textit{Math. Notes},
 Vol.~34, \href{https://doi.org/10.1515/9780691223803}{Princeton University Press}, Princeton, NJ, 1988.

\bibitem{Mikhailov:2012id}
Mikhailov A., Cornering the unphysical vertex, \href{https://doi.org/10.1007/JHEP11(2012)082}{\textit{J.~High Energy Phys.}}
 \textbf{2012} (2012), no.~11, 082, 31~pages, \href{https://arxiv.org/abs/1203.0677}{arXiv:1203.0677}.

\bibitem{Mikhailov:2017mdo}
Mikhailov A., A minimalistic pure spinor sigma-model in {A}d{S},
 \href{https://doi.org/10.1007/jhep07(2018)155}{\textit{J.~High Energy Phys.}} \textbf{2018} (2018), no.~7, 155, 29~pages,
 \href{https://arxiv.org/abs/1706.08158}{arXiv:1706.08158}.

\bibitem{Suszek:2018bvx}
Suszek R.R., Equivariant {C}artan--{E}ilenberg supergerbes for the
 {G}reen--{S}chwarz superbranes~{III}. {T}he wrapping anomaly and the
 super-{${\rm AdS}_5\times\mathbb{S}^5$} background, \href{https://arxiv.org/abs/1808.04470}{arXiv:1808.04470}.

\bibitem{Suszek:2019cum}
Suszek R.R., Equivariant {C}artan--{E}ilenberg supergerbes~{II}. {E}quivariance
 in the super-{M}inkowskian setting, \href{https://arxiv.org/abs/1905.05235}{arXiv:1905.05235}.

\bibitem{Witten:1991mm}
Witten E., On holomorphic factorization of {WZW} and coset models,
 \href{https://doi.org/10.1007/BF02099196}{\textit{Comm. Math. Phys.}} \textbf{144} (1992), 189--212.

\bibitem{Zavaleta:2019dop}
Zavaleta D., Pure spinor string and generalized geometry, \href{https://arxiv.org/abs/1906.05784}{arXiv:1906.05784}.

\end{thebibliography}
\end{document}